\documentclass[11pt]{article}
\usepackage{amsmath,amssymb,bbm,amsthm}
\usepackage{fullpage}
\usepackage{thm-restate,color,xcolor,xspace}
\usepackage{hyperref,cleveref}
\usepackage{graphicx}
\usepackage{algorithm}
\usepackage[noend]{algpseudocode}
\usepackage{enumerate}

\newtheorem{theorem}{Theorem}[section]
\newtheorem{lemma}[theorem]{Lemma}

\newtheorem{observation}[theorem]{Observation}
\newtheorem{claim}[theorem]{Claim}

\newtheorem{remark}[theorem]{Remark}
\newtheorem{algo}[theorem]{Algorithm}
\newcommand{\remove}[1]{}

\makeatletter
  \let\@@Statex\Statex
  \renewcommand\Statex[1][0]{%
    \setlength\@tempdima{\algorithmicindent}%
    \@@Statex\hskip\dimexpr#1\@tempdima\relax}
\makeatother

\begin{document}

\title{Faster Negative-Weight Shortest Paths and\\ Directed Low-Diameter Decompositions}
\author{
Jason Li\footnote{Carnegie Mellon University. email: \tt jmli@cs.cmu.edu} \and
Connor Mowry\footnote{Carnegie Mellon University. email: \tt cmowry@andrew.cmu.edu} \and
Satish Rao\footnote{UC Berkeley. email: \tt satishr@berkeley.edu} \and
}
\date{\today}
\maketitle

\begin{abstract}
We present a faster algorithm for low-diameter decompositions on directed graphs, matching the $O(\log n\log\log n)$ loss factor from Bringmann, Fischer, Haeupler, and Latypov (ICALP 2025) and improving the running time to $O((m+n\log\log n)\log n\log\log n)$ in expectation. We then apply our faster low-diameter decomposition to obtain an algorithm for negative-weight single source shortest paths on integer-weighted graphs in $O((m+n\log\log n)\log(nW)\log n\log\log n)$ time, a nearly log-factor improvement over the algorithm of Bringmann, Cassis, and Fischer (FOCS 2023).

\end{abstract}

\section{Introduction}

We consider the problem of computing shortest paths in a graph
$G=(V,E,w)$ with vertices $V$, edges $E$ and weights $w:E \rightarrow \mathbb Z$. That is, the edge weights are integral but could be negative.  In particular, we give an algorithm that runs in $O((m+n\log\log
n)\log(nW)\log n\log\log n)$ for a graph with $m$ edges, $n$ vertices
and where $W$ bounds the maximum absolute value of a negative edge
weight. This is a nearly log-factor improvement over the
$O((m+n\log\log n)\log(nW)\log^2n)$ time algorithm of Bringmann, Cassis
and Fischer~\cite{bringmann2023negative}. Their work, in turn, improved
on the breakthrough work of Bernstein, Nanongkai, and Wulff-Nilsen~\cite{bernstein2025negative} that gave a near linear $O(m \log^8 n \log W)$ time
algorithm, which itself improved over decades-old previous algorithms
\cite{gabow1989faster,goldberg1995scaling} where the runtime incurred a  polynomial factor of  $\sqrt{n}$ over linear time.

A main tool in these recent developments is a (probabilistic) low-diameter
decomposition (LDD) on a directed graph: given a diameter parameter
$\Delta$, delete a random subset of edges such that (1)~each strongly
connected component in the remaining graph has (weak) diameter at most
$\Delta$, and (2)~each edge is deleted  (cut) with probability at most
$\frac{\ell(n)}{\Delta}$ times the weight of the edge, where $\ell(n)$ is a
loss factor that should be minimized.

In this paper, we give an algorithm to find an LDD with loss $\ell(n) = O(\log
n\log\log n)$ in expected time $O((m+n\log\log n)\log n\log\log n)$,
improving upon the previous algorithm of \cite{bringmann2025near} by three
log-factors in terms of runtime.\footnote{They state their running time as $O(m\log^5n\log\log n)$ but note that one log-factor can be trivially shaved, so we treat their number of log-factors as four.}

This type of decomposition was introduced by Awerbuch
\cite{awerbuch1989network} and has a long history of use for algorithms on
undirected graphs
\cite{bartal1996probabilistic,fakcharoenphol2003tight,calinescu2005approximation,charikar1998approximating}. Related decompositions without
the per-edge probabilistic bounds on being cut were previously used in
approximation algorithms as well in both directed and undirected
graphs. See, for example,
\cite{seymour1995packing,leighton1999multicommodity,klein1997approximation,even1998approximating}

We note the recent $m^{1 + o(1)}$
algorithm of \cite{chen2025maximum} for the minimum cost flow problem
 can also be used to solve the shortest path problem.
This result follows from a long series of work combining convex optimization,
combinatorial pre-conditioning, and dynamic algorithms. We refer the
reader to \cite{chen2025maximum} for a history of this amazing set of
results and techniques.

We also mention that an alternate frame for
strongly polynomial time algorithms, where the runtime only depends on
$n$ and $m$ and not on the size of the edge weights.  Here too, there
have been staggering breakthroughs;  Fineman \cite{fineman2024single}
improved over the classical bound of Bellman-Ford of $O(mn)$ time by
giving a randomized algorithm that runs in time
$\tilde{O}(mn^{7/8})$ and  Fineman's ideas were then used by Huang, Jin and Quanrud
\cite{huang2025faster} to give an $\tilde{O}(mn^{4/5})$ algorithm.\footnote{$\tilde O(\cdot)$ notation hides factors polylogarithmic in $n$.}

Back to the results in this paper, we also replace the noisy binary search procedure from
\cite{bringmann2023negative} that finds a negative cycle when the
shortest path problem is infeasible by directly including logic in the
algorithms themselves to address this case.

We proceed with a discussion of techniques that are used
in previous algorithms and ours. The discussion aims to explain
the shortest path algorithms, the relationship to low diameter decompositions, and how algorithms for both are combined.

\subsection{Techniques}



In the following, we sketch the ideas used in the relevant algorithms and in ours. At times, we sacrifice mathematical precision to communicate the central ideas.


\paragraph{Connecting Negative Shortest Paths to Low Diameter Decompositions.}

If one can compute a shortest path from an added source vertex with a
zero weight arc to all the vertices, one can compute non-negative weights
that preserve all shortest paths. In particular, a {\em price} function, $\phi$, for an edge $e = (u,v)$, has {\em reduced cost} $w_\phi (e) = w(e) - \phi(v) + \phi(u)$.  For shortest
path labels, we have $d(v) \leq d(u) + w(u,v)$ and thus using $d$ as a price function ensures that every edge is non-negative.  The algorithms cannot find such a price function in one go, so they repeatedly compute weaker price functions to reduce the most negative edges. 

For the sake of this discussion, the central structure to think about are  negative weight
paths. And also recall that the problem is only feasible if there is
no negative cycle which we assume for the following description.

In this context, Bernstein, Nanongkai, and Wulff-Nilsen
\cite{bernstein2025negative} consider a graph, $G_{\geq 0}$, where all
the negative weight edges in $G$ are assigned weight 0.  A path whose
length is $-L$ in $G$ cannot have both endpoints in a strongly
connected component of $G_{\geq 0}$ of diameter less than $L$; otherwise, a negative cycle is formed by the negative path of length $-L$ and the path between the endpoints of
length less than $L$ in the strongly connected component.

Furthermore, in a low diameter decomposition with parameter $\Delta < L$, there is no negative weight path of length less than $-L$ entirely inside any strongly
connected component of the decomposition.  That is, for decompositions
of $G_{\geq 0}$ with smaller parameter $\Delta$, the shortest paths
in each component become smaller which suggests the recursive use
of low diameter decompositions is useful.

To complete the recursion, one needs to stitch together solutions from the strongly connected components in the LDD which suggests that
negative length paths should not cross between strongly connected
components in the LDD too often. Here, the probabilistic bound in part (2) of the LDD definition is useful.

The algorithm of~\cite{bernstein2025negative} requires one more idea based on scaling. They consider
the edges, $E'$, of roughly the maximum negative weight $-W$ and add
$W/2$ to each edge in $E'$ to produce a graph, $G_+$.  A shortest
path labelling in $G_+$ can be used as a price function where the maximum absolute negative
value drops by a factor of $2$.

For a path with $k$ edges in $E'$ which remains negative
in $G_+$, we argue that its endpoints cannot be in the same connected component
for an LDD with parameter $\Delta = kW/2$ as follows. The path in $G$ 
has length $ \leq -kW/2$ as the path is negative in $G_+$ and each
of $k$ edges is more negative by $W/2$ in $G$ and the diameter of
the any SCC is $\Delta = kW/2$ and thus there would be a negative
cycle if the endpoints were to be in the same SCC.

Moreover, the positive length  of such a path is at most $\Delta$ since the entire path is negative in $G_+$.  This latter observation and the
cutting probability bounds from the LDD to ensures  the path does
not cross between strongly connected components too often, in
particular, no more than the loss, $\ell(n)$ of the LDD construction on expectation.

Notice, there are two parameters; the number $k$ of edges in $E'$ and
$W$. They fix $W$ and then recursively decompose $G'$ with LDDs such
that for level $i$, the LDD cuts paths with $k = n/2^i$ edges in $E'$.
Moreover, each such path is only cut polylogarithmic times at each
level on average.

This leads to a near linear time procedure that produces a weight
function that reduces the negativity of edges in $E'$ by half and
preserves all shortest paths. Recursing on the value of $W$ completes
their algorithm at a cost of  $O(\log W)$ factor.

The other factors are the $\log n$ levels of the LDD that deals with
number of edges in $E'$ on a path and the polylogarithmic loss of the LDD, as well as the runtime for computing the LDDs.

Bringmann, Cassis and Fischer \cite{bringmann2023negative} improve this
in a variety of ways including new ideas for constructing Low Diameter
Decompositions as well as integrating the recursions of the LDD
constructions with the recursion for the negative weight shortest
paths. Later on,~\cite{fischer2025simple} simplify the argument by ignoring the parameter $k$ and working with the diameter $\Delta$ directly.

We proceed first by discussing background on constructing low diameter decompositions and then proceed with a discussion of \cite{bringmann2023negative}.

\paragraph{Constructing Low Diameter Decompositions: techniques.}

Informally, low diameter decompositions remove edges (1) to decompose
a graph into low diameter pieces and (2) have each edge have
low probability of being removed. We say an edge is cut if it is removed below.

For this discussion, we assume we have an unweighted graph as that captures
the bulk of the interesting ideas. In this context,
producing a piece in the decomposition is typically done by 
breadth first search of low depth which in undirected graphs produces
a low diameter piece. The piece may not have low diameter strongly connected components in directed
graphs and thus all current methods recurse on the pieces until they do indeed satisfy the low diameter condition. 

We proceed, for now, by sketching techniques developed for undirected graphs.
The frame is typically do a truncated shortest path say truncated at
$\Delta$, remove the piece and continue. This kind of procedure can produce
cuts of small size, intuitively, $\tilde{O}(\frac{|E|}{\Delta})$, as one
can choose to cut in $\Theta(\Delta)$ possible depths and choose according
to the fewest edges. Or can do this probabilistically by choosing a random
depth among the $\Delta$ choices, which then has the property that each
edge is cut with probability, intuitively, proportional to
$\frac{1}{\Delta}$.

The flaw in this reasoning is that when one cuts another piece from
the remaining graph each remaining edge again suffers a risk of being
cut. Naively, this could lead to extra factor of $n$.  Choosing the
depth  of the breadth first search from an exponential distribution fixes this as when an edge suffers the
possibility of being cut it also has a good probability of being
inside the piece (``saved'' from further risk) by the memoryless
property of the exponential distribution. There is an overhead of
$O(\log n)$ introduced here as the depth may become $\Omega(\Delta \log n)$ when choosing from an exponential distribution with mean $\Delta$.

An alternative idea from \cite{calinescu2005approximation} in the context of approximation
algorithms, was to choose starting vertices from a random permutation
as well as a random diameter in a range of width $\Theta(\Delta)$
preserving the diameter bound. Here an edge can only be cut by the $i$th closest vertex if that vertex is first in the random permutation as the closer vertices would save the edge in this iteration. This happens with probability $1/i$. This leads to a possibly penalty of $\sum_{i}
1/i = O(\ln n)$ over $O(\frac{1}{\Delta})$ on the probability of
being cut.

This idea not only works for fixed $\Delta$, but can be used to produce
hierarchical decompositions with $O(\log n)$ loss \cite{fakcharoenphol2003tight}. This is
important for directed decompositions as one needs to recurse as noted above and also for
our recursive negative weight shortest path algorithm.

To deal with the recursion in finding LDD's for directed graphs, Seymour \cite{seymour1995packing}
observed that the $O(\log n)$ penalty on cut size for low depth breadth first search could be varied
according to the size of the piece that is cut off. That is, one can
pay very little overhead to cut a large piece off and more to cut off
a small piece. But for a recursion based on piece size, the small
pieces need to be recursed on less.  Seymour optimized this tradeoff to
give a low diameter deterministic decomposition with overhead $O(\log
n \log \log n)$. While this does bound the total number of edges that
are cut, it does not give a per edge bound on the probability of being
in the cut. The runtime is also polynomial, not near linear.

We note that Seymour's ideas were used to produce a hierarchical
decompositions of undirected graphs with overheads of $O(\log
n\log\log n)$ over all the levels of recursion as well. See for
example, \cite{bartal1998approximating}.

\paragraph{Entwining recurrences for LDD and Negative Weight Shortest Path.}

Recall that there is a recurrence for negative weight shortest paths
that entails doing low diameter decompositions for a geometric
sequence of diameters, $\Delta$.  Moreover, for a low diameter decomposition
in directed graphs, algorithms (currently) recurse for even one
diameter.

Bringmann, Cassis and Fischer \cite{bringmann2023negative} combine
the recurrences for shortest paths and low diameter decompositions in an ingenious but simple manner. At a particular
level of decomposition, they find pieces that either have diameter
$\Delta$ or are smaller by a constant factor and  where each edge has probability
$O(\frac{\log n}{\Delta})$ of being cut.  Then they recursively compute price
functions on each piece to ensure that all internal edge weights are
positive in $G_+$. Given that any negative weight shortest path on
expectation has $\ell(n) = O(\log n)$ crossings in this decomposition, they use expected time proportional to the time of $\ell(n)$ Dijkstra computations to complete the algorithm.

They build on the ``ball-growing'' (which is a directed version of
breadth first search) procedures of \cite{bernstein2025negative} to do
so efficiently. They also use the exponential distribution to pick
ball radii and use sampling techniques to estimate the
sizes of balls for all vertices. 

Since either the size drops or the length drops, the number of
recursive levels is $O(\log n)$. Combined with the loss factor, the total overhead becomes $O(\log^2 n)$, and then
the scaling approach of halving the most negative weight edge adds another $O(\log(nW))$ factor.

Note they do not compute a complete LDD for any level. Though, one
such procedure is developed by Bringmann, Fischer, Haeupler and Latypov in
\cite{bringmann2025near} with near optimal $\ell(n)$ of $O(\log n \log
\log n)$ but runtime that is a factor of $O(\log^4 n)$ more than
linear. They use a version of random permutations from \cite{calinescu2005approximation}
along with a sequence of samples that we use as well.

\paragraph{This paper.}

We improve the algorithm in \cite{bringmann2023negative} that produces a price function whose reduced costs are non-negative for $G_+$  from having $O(\log^2 n)$ expected time to $O(\log n \log\log n)$ expected time.

To review, the shortest path algorithm proceeds by decomposing the
current graph into pieces where the diameter of pieces, $\Delta$, is significantly reduced or the size is significantly reduced and the cost of that decomposition is
the probability that an edge gets cut is $O(\frac{\ell(n)}{\Delta})$.  That
probability corresponds to the cleanup work to find a global price
function at that level. The total work (disregarding the computation
time for the decomposition) is roughly $O(\ell(n) \log n)$ over 
$\log n$ levels of recursion to deal with the number of edges in a path. And in \cite{bringmann2023negative},
$\ell(n)$ is $O(\log n)$.

As done by Seymour \cite{seymour1995packing}, the cutting procedure can be adjusted
so that the ``overhead'' is larger on small pieces. In particular, if
at a level all the pieces are of size say $\Theta(m/2^s)$ the
probability of an edge being cut can be made (proportional to)
$O(s)$.\footnote{ The factor of proportionality is $O(\log \log n)$
which we leave out in this discussion.} Moreover, for pieces of size $m/2^s$ one only need recurse $\log n - s$ more
times.  Thus, the total overhead over all levels is (proportional to)
$O(\log n)$ in this case. 

Thus, the algorithm proceeds at each level to determine whether one
can divide the graph into pieces of size $m/2^s$ or divide into
possibly large pieces of diameter $L/2$.  To be sure, the sizes are
grouped according to sizes that are between $m/2^{2^i}$ and $m/2^{2^{i+1}}$  as the cutting overhead for pieces  inside a group are within a constant factor:  the penalty is proportional to $\log (2^{2^i}) \log\log n$
and thus the penalty of $2^i \log \log n  = \Theta((2^{i+1}) \log\log n)$ is constant
for sizes in the group determined by $i$.
 Moreover, with this parameterization, the number of sizes is $O(\log \log n)$ so one can iterate over them.

We do this by iterating over the possible size groups from largest size to smallest. We compute shortest
paths of limited depth from each sample of increasing size over the iterations.  The failure of one iteration ensures that a typical ball is small. Thus, the work in the next iteration can be bounded with a larger sample of sources as the number of nodes touched even from many sources in the larger sample remains small. 

This is a main thread of the algorithm and its improvement
but there is technical work to make it precise and complete.
Moreover, we need  ideas from
\cite{calinescu2005approximation} to get overhead bounds and from
\cite{bringmann2025near} to obtain efficient runtimes. There is also technical work needed to get good success probability without
oversampling.  We leave that to the technical sections.

\subsection{Preliminaries and Paper Structure}
Let $G=(V,E,w)$ be a directed graph with non-negative and integral edge weights. The degree $\deg(v)$ of a vertex $v\in V$ is the number of edges incident to $v$ (which does not depend on edge directions or weights), and the volume $\textup{vol}(U)$ of a vertex subset $U\subseteq V$ equals $\sum_{v\in U}\deg(v)$. For a vertex subset $U\subseteq V$, define $G[U]$ as the induced graph on $U$, and define $E[U]$ as the edges of this induced graph. For vertices $u,v\in V$, define $d(u,v)$ to be the distance from $u$ to $v$ according to the edge weights. For a vertex subset $U\subseteq V$, define the \emph{weak diameter} of $U$ as $\max\{d(u,v):u,v\in U\}$.

The paper is divided into two main sections. \Cref{sec:1} describes our new directed LDD algorithm, and \Cref{sec:sssp-main} applies it to speed up negative-weight shortest paths. \Cref{sec:1} is fully self-contained, so a reader primarily interested in directed LDD may skip \Cref{sec:sssp-main}. For a reader mainly interested in negative-weight shortest paths, we package the necessary guarantees of our LDD algorithm in \Cref{sec:decomposition-interface}, so that the rest of \Cref{sec:1} can be skipped.

We defer additional preliminaries for negative-weight shortest paths to its own section; see \Cref{sec:sssp-prelims}.

\section{Directed LDD Algorithm}\label{sec:1}

In this section, we present our improved algorithm for low-diameter decomposition (LDD) on directed graphs. We begin with the main theorem statement and a technical overview.

\begin{theorem}
There is a randomized algorithm for directed low-diameter decomposition that achieves loss factor $O(\log n\log\log n)$, runs in $O((m+n\log\log n)\log n\log\log n)$ expected time, and succeeds with high probability.\footnote{\emph{With high probability} means with probability at least $1-1/n^C$ for arbitrarily large constant $C>0$.}
\end{theorem}

Our algorithm follows the same template as~\cite{bernstein2025negative,bringmann2025near}: compute an out-ball or in-ball $B$, remove the edges cut by the ball, and recursively solve the induced graphs $G[B]$ and $G[V\setminus B]$. Observe that by removing all edges cut by $B$, the vertices in $B$ can no longer share strongly connected components with the vertices in $V\setminus B$. Therefore, if the algorithm recursively computes directed LDDs in $G[B]$ and $G[V\setminus B]$, which means deleting edges such that each strongly connected component has weak diameter at most $\Delta$, then deleting these same edges, along with the edges cut by the ball $B$, also produces strongly connected components with weak diameter at most $\Delta$ in the original graph $G$.

To bound the recursion depth and speed up the running time, the algorithms actually compute multiple balls in each recursive instance so that each recursive component has sufficiently fewer edges.

The algorithm of~\cite{bringmann2025near} adopts a specialized recursive structure colloquially known as Seymour's trick~\cite{seymour1995packing}, and we use the same general approach. At a high level, our algorithm has two main differences.

\paragraph{CKR instead of geometric ball-growing.}
The algorithm of~\cite{bringmann2025near} uses the geometric ball-growing technique: sample a random radius from a geometric distribution, and cut the appropriate out-ball or in-ball. The main benefit of the geometric distribution is its \emph{memoryless} property, which is crucial in bounding edge cut probabilities. On the other hand, to avoid losing another factor of $O(\log n)$ in the bound of $\ell(n)$, the parameters are set in a way that the probability that the radius grows too large is larger than $1/\textup{poly}(n)$, so it cannot be ignored with high probability and forces a delicate probability analysis.

Instead, we adopt an approach first used by Calinescu, Karloff, and Rabani~\cite{cualinescu2000improved} in the context of LP rounding for multi-way cut. Here, the radius is no longer geometric, but uniformly random in a given range, and each ball uses the same sampled radius. The twist is that CKR computes the balls in a \emph{random} order, which is key to their analysis.

Similar to \cite{bringmann2025near}, we cannot afford to grow balls from all (remaining) vertices, which would incur a factor of $O(\log n)$ in the bound of $\ell(n)$, so we also \emph{sample} a subset of vertices from which to grow balls. On each application of CKR, for a given edge $(u,v)$, the cut probability has a factor that is logarithmic in the number of sampled vertices whose balls can possibly cut $(u,v)$. This factor is carefully balanced with the sizes of the recursive instances by an application of Seymour's trick~\cite{seymour1995packing}.

\paragraph{Elimination of heavy vertices.}
In the procedure of cutting  balls, it is possible that the ball itself contains nearly all the vertices. In this case, the recursion depth can be large, affecting both the cut probability and running time. Our second contribution is a preprocessing step in $O((m+n\log\log n)\log\log n)$ expected time that reduces to the special case where all out-balls and in-balls contain at most 75\% of the edges. We remark that the threshold of 75\% is arbitrary and can be replaced by any constant, even less than 50\%. We expect this reduction to be useful in future developments, since it reduces the general problem to a setting that is easier to reason about.

The main difficulty in eliminating heavy vertices is that existing algorithms to even \emph{detect} heavy vertices run in $O((m+n\log\log n)\log n)$ time~\cite{bringmann2023negative}, by running single-source shortest paths from $O(\log n)$ sampled vertices. We cannot afford a logarithmic overhead on this step, since our recursion depth may also be logarithmic. Instead, we run shortest paths from only $O(\log\log n)$ sampled vertices, which may incorrectly label each vertex as light or heavy with probability $1/\text{polylog}(n)$. If the algorithm attempts to grow a ball from an incorrectly labeled light vertex, and discovers that it is actually heavy, then the algorithm simply gives up and restarts. Fortunately, we only try to grow balls from $O(\log n)$ vertices, so the overall probability of restarting is low.

The remaining details of the algorithm are technical, splitting into cases depending on whether the majority of vertices (by volume) are in-heavy, and whether the majority are out-heavy.



\subsection{The Main Algorithm}

Let $G=(V,E,w)$ be the input graph on $n$ vertices and $m$ edges, and let $\Delta$ be the diameter parameter. For vertex $t\in V$ and non-negative number $r$, define the \emph{out-ball} $B^+_G(t,r)=\{v\in V:d(t,v)\le r\}$ and \emph{in-ball} $B^-_G(t,r)=\{v\in V:d(v,t)\le r\}$. We say that an edge $(u,v)$ is \emph{cut} by out-ball $B^+_G(t,r)$ if $u\in B^+_G(t,r)$ and $v\notin B^+_G(t,r)$, and it is \emph{cut} by in-ball $B^-_G(t,r)$ if $v\in B^-_G(t,r)$ and $u\notin B^-_G(t,r)$.

Call a vertex \emph{out-light} in subgraph $H\subseteq G$ if $|E[B^+_H(v,\Delta/8)]|\le\frac34m$, and \emph{in-light} if\linebreak $|E[B^-_H(v,\Delta/8)]|\le\frac34m$. The algorithm begins with a preprocessing step where $V$ is replaced by a subset $U\subseteq V$ of vertices that are both out-light and in-light. This preprocessing step is technical, so we only state its guarantees in this section, deferring the algorithm to \Cref{sec:2}.

\begin{restatable}[Preprocessing]{theorem}{Preprocessing}\label{thm:2.2}
Let $G=(V,E,w)$ be the input graph on $n$ vertices and $m$ edges, let $\Delta$ be the diameter parameter, and let $\epsilon$ be the error parameter. There is an algorithm that runs in $O((m+n\log\log n)\log\log(1/\epsilon))$ expected time and cuts a random subset $S\subseteq E$ of edges such that
 \begin{enumerate}
 \item Each edge joins $S$ with probability at most $O(\frac{\log\log(1/\epsilon)}\Delta)$.
 \item The algorithm correctly labels each strongly connected component of $G-S$ by one of the following:\label{item:2}
  \begin{enumerate}
  \item The component has at most $\frac34m$ edges.\label{item:a}
  \item The component has weak diameter at most $\Delta$\label{item:b}.
  \item All vertices in the component are both out-light and in-light in $G[U]$, where $U$ is the union of all components with this label. This case is only correct with probability at least $1-2m\epsilon$.\label{item:c}
  \end{enumerate}
 \end{enumerate}
\end{restatable}

Given this preprocessing step, we now describe the main algorithm. At a high level, the algorithm is recursive, cutting out-balls and in-balls of a random radius for a number of iterations. The algorithm is simple enough to state completely, so we present it in pseudocode below.
\begin{algo} \label{algo:1} Let $G=(V,E,w)$ be the input graph on $n$ vertices and $m$ edges, let $\Delta$ be the diameter parameter, and let $\epsilon$ be the error parameter.
\begin{enumerate}
\item Run the preprocessing step (\Cref{thm:2.2}) and recursively solve all components labeled~(\ref{item:a}). Ignore all components labeled~(\ref{item:b}), since they have the target weak diameter. Let $U\subseteq V$ be the union of all components labeled~(\ref{item:c}), which we now process. Throughout the algorithm, $U$ is the set of remaining vertices.\label{step:1}
\item Let $L=\lceil\lg\lg m\rceil$, and define the decreasing sequence $a_0=\Delta/8$ and $a_i=a_{i-1}-\frac\Delta{16\min\{L,\,2^i\}}$ for $i\in\{1,2,\ldots,L\}$.
\item For $i\in\{1,2,\ldots,L\}$ in increasing order:\label{item:for}
 \begin{enumerate}
 \item Let $U_i$ be a snapshot of the current set $U$ of remaining vertices.
 \item Sample each vertex $v\in U_i$ with probability $\min\{1,\, 2^{2^i}\ln(1/\epsilon)\cdot\frac{\deg(v)}{2m}\}$. Let $S_i$ be the random sample.\label{step:a}
 \item Sample a random radius $r_i\in[a_i,a_{i-1}]$.
 \item For each vertex $v\in S_i$ in a random order:\label{step:c}
  \begin{enumerate}
  \item Cut the out-ball $B^+_{G[U_i]}(v,r_i)$ from the remaining graph $G[U]$. Recursively solve\linebreak $G[B^+_{G[U_i]}(v,r_i)\cap U]$, and then set $U\gets U\setminus B^+_{G[U_i]}(v,r_i)$.
  \item Cut the in-ball $B^-_{G[U_i]}(v,r_i)$ from the remaining graph $G[U]$. Recursively solve\linebreak $G[B^-_{G[U_i]}(v,r_i)\cap U]$, and then set $U\gets U\setminus B^-_{G[U_i]}(v,r_i)$.
  \end{enumerate}
 \end{enumerate}
\end{enumerate}
\end{algo}
\begin{remark}
The sequence $a_0,a_1,\ldots,a_L$ is non-negative since
\[ \sum_{i=1}^L\frac\Delta{16\min\{L,\,2^i\}}\le\sum_{i=1}^L\left(\frac\Delta{16L}+\frac\Delta{16\cdot2^i}\right)=\frac1{16}\bigg(\sum_{i=1}^L\frac\Delta{L}+\sum_{i=1}^L\frac\Delta{2^i}\bigg)\le\frac1{16}(\Delta+\Delta)=\frac\Delta8=a_0 ,\]
so the algorithm is computing balls of non-negative radius.
\end{remark}

\begin{lemma}\label{lem:2}
At the beginning of each iteration $i>1$, with probability at least $1-2m\epsilon$, each vertex $u\in U_i$ satisfies $\textup{vol}(B^\pm_{G[U_i]}(u,a_{i-1})\cap U_i)\le 2m/2^{2^{i-1}}$.
\end{lemma}
\begin{proof}
For simplicity, we only consider the case of in-balls.  Consider a vertex $u\in U_{i-1}$ with\linebreak $\textup{vol}(B^-_{G[U_{i-1}]}(u,a_{i-1})\cap U_{i-1})>2m/2^{2^{i-1}}$ at the beginning of iteration $i-1$. Our goal is to show that $u$ is almost certainly removed from $U$ on iteration $i-1$. On that iteration, each vertex in $U_{i-1}$ is sampled with probability $\min\{1,\, 2^{2^{i-1}}\ln(1/\epsilon)\cdot\frac{\deg(v)}{2m}\}$, so the probability of sampling no vertex in $B^-_{G[U_{i-1}]}(u,a_{i-1})\cap U_{i-1}$ is
\begin{align*}
&\prod_{v\in B^-_{G[U_{i-1}]}(u,a_{i-1})\cap U_{i-1}}\left(1-\min\left\{1,\, 2^{2^{i-1}}\ln(1/\epsilon)\cdot\frac{\deg(v)}{2m}\right\}\right)
\\\le{}&\prod_{v\in B^-_{G[U_{i-1}]}(u,a_{i-1})\cap U_{i-1}}\exp\left(- 2^{2^{i-1}}\ln(1/\epsilon)\cdot\frac{\deg(v)}{2m}\right)
\\={}&\exp\left(- 2^{2^{i-1}}\ln(1/\epsilon)\cdot\frac{\textup{vol}(B^-_{G[U_{i-1}]}(u,a_{i-1})\cap U_{i-1})}{2m}\right)
\\\le{}&\exp\bigg(- 2^{2^{i-1}}\ln(1/\epsilon)\cdot\frac{2m/2^{2^{i-1}}}{2m}\bigg)
\\={}&\exp(-\ln(1/\epsilon))
\\={}&\epsilon.
\end{align*}
That is, with probability at least $1-\epsilon$, at least one vertex $v\in B^-_{G[U_{i-1}]}(u,a_{i-1})$ is sampled on iteration $i-1$. In this case, since $r_{i-1}\ge a_{i-1}$, we also have $v\in B^-_{G[U_{i-1}]}(u,r_{i-1})$, or equivalently $u\in B^+_{G[U_{i-1}]}(v,r_{i-1})$, so vertex $u$ is removed from $U$ when cutting $B^+_{G[U_{i-1}]}(v,r_{i-1})$ (if not earlier). By a union bound over all $u\in U_{i-1}$, with probability at least $1-m\epsilon$ the same holds for all $u\in U_{i-1}$ with $\textup{vol}(B^-_{G[U_{i-1}]}(u,a_{i-1})\cap U_{i-1})\ge 2m/2^{2^{i-1}}$. In this case, the only vertices that remain in $U$ after iteration $i-1$ must have $\textup{vol}(B^-_{G[U_{i-1}]}(u,a_{i-1})\cap U_{i-1})\le 2m/2^{2^{i-1}}$. Since $U_i\subseteq U_{i-1}$, we also have $\textup{vol}(B^-_{G[U_i]}(u,a_{i-1})\cap U_i)\le 2m/2^{2^{i-1}}$, as promised.

The proof for out-balls $\textup{vol}(B^+_{G[U_i]}(u,a_{i-1})\cap U_i)\le 2m/2^{2^{i-1}}$ is symmetric, and we union bound over both cases for a final probability of $1-2m\epsilon$.
\end{proof}
\begin{lemma}\label{lem:4}
Consider an edge $(u,v)$ with $u,v\in U_i$ at the beginning of iteration $i$. Then, edge $(u,v)$ is cut on iteration $i$ with probability $O(\frac{2^i\log\log(1/\epsilon)}{\Delta})\cdot w(u,v)$.
\end{lemma}
\begin{proof}
For simplicity, we only consider out-balls $B^+_{G[U_i]}(t,r_i)$.  Let random variable $k$ be the number of sampled vertices in $U_i$ at the beginning of iteration $i$. The expected value of $k$ is
\begin{align*}
\sum_{v\in U_i}\min\left\{1,\, 2^{2^i}\ln(1/\epsilon)\cdot\frac{\deg(v)}{2m}\right\}&\le 2^{2^i}\ln(1/\epsilon)\cdot\frac{\textup{vol}( U)}{2m}
\\&\le 2^{2^i}\ln(1/\epsilon)\cdot\frac{\textup{vol}(V)}{2m}
\\&= 2^{2^i}\ln(1/\epsilon). 
\end{align*}

Given the value of $k$, order the sampled vertices $t_1,\ldots,t_k$ by value of $d(t_j,u)$, with ties broken arbitrarily. Recall that the algorithm iterates over $t_1,\ldots,t_k$ in a random order. In order for ball $B^+_{G[U_i]}(t_j,r_i)$ to cut edge $(u,v)$, two conditions must hold:
 \begin{enumerate}
 \item The radius $r_i$ must satisfy $d(t_j,u)\le r_i<d(t_j,u)+w(u,v)$, since this holds whenever $u\in B^+_{G[U_i]}(t_j,r_i)$ and $v\notin B^+_{G[U_i]}(t_j,r_i)$. There are $a_{i-1}-a_i=\frac\Delta{16\min\{L,\,2^i\}}$ many choices of $r_i$, so this happens with probability at most $\frac{16\min\{L,\,2^i\}}{\Delta}\cdot w(u,v)$.
 \item The vertex $t_j$ must precede $t_1,\ldots,t_{j-1}$ in the ordering, since otherwise another ball would claim $u$ first. Since the ordering is random, this happens with probability $1/j$.
 \end{enumerate}
The two events above are independent, so the probability that ball $B^+_{G[U_i]}(t_j,r_i)$ cuts edge $(u,v)$ is at most $\frac{16\min\{L,\,2^i\}}\Delta\cdot w(u,v)\cdot\frac1j$. Summing over all $j$, the overall probability that edge $(u,v)$ is cut by a ball $B^+_{G[U_i]}(t_j,r_i)$ is at most
\[ \frac{16\min\{L,\,2^i\}}\Delta\cdot\left(\frac11+\frac12+\cdots+\frac1k\right)\cdot w(u,v)\le \frac{16\min\{L,\,2^i\}}\Delta\cdot(\ln k+1)\cdot w(u,v) .\]
Taking the expectation over $k$, and using the fact that $\ln k$ is concave, the overall probability is at most
\begin{align*}
&\frac{16\min\{L,\,2^i\}}\Delta\cdot(\mathop{\mathbb E}[\ln k]+1)\cdot w(u,v)
\\\le{}&\frac{16\min\{L,\,2^i\}}\Delta\cdot(\ln\mathop{\mathbb E}[k]+1)\cdot w(u,v)
\\\le{}&\frac{16\min\{L,\,2^i\}}\Delta\cdot(\ln(2^{2^i}\ln(1/\epsilon))+1)\cdot w(u,v)
\\\le{}&O\left(\frac{\min\{\log\log(1/\epsilon),\,2^i\}\cdot(2^i+\log\log(1/\epsilon))}{\Delta}\right)\cdot w(u,v)
\\\le{}&O\left(\frac{2^i\log\log(1/\epsilon)}\Delta\right)\cdot w(u,v) ,
\end{align*}
where the last step uses the inequality $\min\{a,b\}\cdot(a+b)\le\min\{a,b\}\cdot2\max\{a,b\}=2ab$.

Finally, the analysis is symmetric for in-balls, so the overall cut probability at most doubles.
\end{proof}

\begin{lemma}\label{lem:5}
Over the entire recursive algorithm, each edge $(u,v)$ is cut with probability $O(\frac{\log m\log\log(1/\epsilon)}\Delta)\cdot w(u,v)+O(m^3\epsilon)$.
\end{lemma}
\begin{proof}
By \Cref{thm:2.2}, each edge $(u,v)$ is cut with probability $O(\frac{\log\log(1/\epsilon)}\Delta)\cdot w(u,v)$ in the preprocessing step, and since the recursion depth is $O(\log m)$, the total probability of being cut by preprocessing is $O(\frac{\log m\log\log(1/\epsilon)}\Delta)\cdot w(u,v)$. For the rest of the proof, we focus on the probability of cutting an edge after preprocessing.

Consider a function $f(m)$ such that any edge $(u,v)$ is cut with probability at most $f(m)\cdot w(u,v)$ on any graph with at most $m$ edges (after preprocessing). Consider an instance of the recursive algorithm. For each $i\in\{1,2,\ldots,L\}$, let $p_i$ be the probability that $u,v\in U_i$. For each $i$, by \Cref{lem:4}, a fixed edge $(u,v)$ is cut on iteration $i$ with probability at most $p_i\cdot\frac{C\cdot2^i\log\log(1/\epsilon)}{\Delta}\cdot w(u,v)$ for some constant $C>0$. By the definition of $p_i$, the probability that $u,v\in U_i$ holds and $u,v\in U_{i+1}$ does \emph{not} hold is exactly $p_i-p_{i+1}$, where we define $p_{L+1}=0$. (Note that all vertices leave $U$ after iteration $L$ since every vertex is sampled on that iteration.) In particular, with probability at most $p_i-p_{i+1}$, both $u$ and $v$ are cut out by some common ball $B^\pm_{G[U_i]}$ on iteration $i$, which opens the possibility that the edge is cut recursively. If $i=1$, then \Cref{thm:2.2} holds with probability at least $1-2m\epsilon$, in which case the recursive instance $G[B^\pm_{G[U_i]}\cap U]$ has at most $\frac34m=1.5m/2^{2^{i-1}}$ edges since all vertices are out-light and in-light. If $i>1$, then \Cref{lem:2} holds on iteration $i$ with probability at least $1-2m\epsilon$, in which case $|E[B^\pm_{G[U_i]}\cap U]|\le\textup{vol}(B^\pm_{G[U_i]}\cap U)/2\le m/2^{2^{i-1}}\le1.5m/2^{2^{i-1}}$. Summing over all $i$, and ignoring the $2m\epsilon$ probabilities of failure for now, we can bound the cut probability divided by $w(u,v)$ by
\begin{align*}
&\sum_{i=1}^L\left( p_i\cdot\frac{C\cdot2^i\log\log(1/\epsilon)}{\Delta}+(p_i-p_{i+1})f(1.5m/2^{2^{i-1}})\right)
\\={}&\sum_{i=1}^L\left(\big((p_i-p_{i+1})+(p_{i+1}-p_{i+2})+\cdots+(p_L-p_{L+1})\big)\cdot\frac{C\cdot2^i\log\log(1/\epsilon)}{\Delta}+(p_i-p_{i+1})f(1.5m/2^{2^{i-1}})\right)
\\={}&\sum_{i=1}^L(p_i-p_{i+1})\cdot\left(\frac{C\cdot(2^1+2^2+\cdots+2^i)\log\log(1/\epsilon)}{\Delta}+f(1.5m/2^{2^{i-1}})\right)
\\\le{}&\sum_{i=1}^L(p_i-p_{i+1})\cdot\left(\frac{C\cdot2^{i+1}\log\log(1/\epsilon)}{\Delta}+f(1.5m/2^{2^{i-1}})\right).
\end{align*}
Since the values $(p_i-p_{i+1})$ sum to $1$, the summation is a convex combination of the terms $\big(\frac{C\cdot2^{i+1}\log\log(1/\epsilon)}{\Delta}+f(1.5m/2^{2^{i-1}})\big)$, so we can bound the summation by the maximum such term. In conclusion, we have the recursive bound
\[ f(m)\le\max_{i=1}^L\left(\frac{C\cdot2^{i+1}\log\log(1/\epsilon)}{\Delta}+f(1.5m/2^{2^{i-1}})\right) .\]
We now solve the recurrence by $f(m)\le\frac{18C\lg m\log\log(1/\epsilon)}\Delta$. By induction on $m$, we bound each term in the maximum above by
\begin{align*}
&\frac{C\cdot2^{i+1}\log\log(1/\epsilon)}{\Delta}+f(1.5m/2^{2^{i-1}})
\\\le{}&\frac{C\cdot2^{i+1}\log\log(1/\epsilon)}{\Delta}+\frac{18C\lg(1.5m/2^{2^{i-1}})\log\log(1/\epsilon)}\Delta
\\={}&\frac{2C\cdot2^i\log\log(1/\epsilon)}{\Delta}+\frac{18C(\lg m+\lg 1.5-2^{i-1})\log\log(1/\epsilon)}\Delta
\\\le{}&\frac{2C\cdot2^i\log\log(1/\epsilon)}{\Delta}+\frac{18C(\lg m-\frac13\cdot2^{i-1})\log\log(1/\epsilon)}\Delta
\\={}&\frac{18C\lg m\log\log(1/\epsilon)}\Delta-\frac{C\cdot2^i\log\log(1/\epsilon)}\Delta,
\end{align*}
and therefore
\begin{align*}
f(m)&\le\max_{i=1}^L\left(\frac{C\cdot2^{i+1}\log\log(1/\epsilon)}{\Delta}+f(1.5m/2^{2^{i-1}})\right)
\\&\le\max_{i=1}^L\left(\frac{18C\lg m\log\log(1/\epsilon)}\Delta-\frac{C\cdot2^i\log\log(1/\epsilon)}\Delta\right)
\\&\le\frac{18C\lg m\log\log(1/\epsilon)}\Delta ,
\end{align*}
completing the induction.

Finally, we consider the probability of failing either \Cref{thm:2.2} or \Cref{lem:2} on any recursive instance. There are at most $m$ recursive instances, each with $L$ iterations, so the overall failure probability is at most $Lm\cdot O(m\epsilon)\le O(m^3\epsilon)$.
\end{proof}

\begin{lemma}\label{lem:6}
Assume that $\epsilon\le m^{-5}$. The algorithm can be implemented in $O((m+n\log\log n)\log n\log\log n)$ expected time on graphs with integral and polynomially-bounded edge weights.
\end{lemma}
\begin{proof}
By \Cref{thm:2.2}, the preprocessing step takes $O((m+n\log\log n)\log\log(1/\epsilon))$ expected time, and since the recursion depth is $O(\log m)$, the total expected running time over all preprocessing is\linebreak $O((m+n\log\log n)\log m\log\log(1/\epsilon))$.

For the rest of the algorithm, the running time is dominated by computing the sets $B^\pm_{G[U_i]}(v,r_i)\cap U$, which requires single-source shortest paths computations. For simplicity, we only consider out-balls $B^+_{G[U_i]}(v,r_i)\cap U$, since the other case is symmetric. To implement this step, we run Dijkstra's algorithm for each $v\in S_i$ in the random ordering with one important speedup. Suppose that we are currently computing $B^+_{G[U_i]}(v,r_i)\cap U$ for some $v\in S_i$ and are about to add a vertex $u\in B^+_{G[U_i]}(v,r_i)$ to the priority queue in Dijkstra's algorithm. If there is a vertex $v'\in S_i$ that precedes $v$ in the random ordering such that $d(v',u)\le d(v,u)$, then the algorithm skips $u$, i.e., do not add it to the priority queue. We now justify why this speedup does not affect our goal of computing $B^+_{G[U_i]}(v,r_i)\cap U$. If the algorithm skips $u$, then for any vertex $t\in B^+_{G[U_i]}(v,r_i)$ whose shortest path from $v$ passes through $u$, there is a shorter path from $v'$ to $t$ (also passing through $u$), so we must have $t\in B^+_{G[U_i]}(v',r_i)$ as well. Since $v'$ precedes $v$ in the random ordering, the vertex $t$ is already removed from $U$ when $v'$ is processed (if not earlier). Hence, our speedup only skips over vertices that no longer belong to $U$.

We now bound the expected number of times a vertex $u$ can be added to Dijkstra's algorithm over all computations of $B^+_{G[U_i]}(v,r_i)\cap U$. For a given iteration $i$, recall from the proof of \Cref{lem:4} that the number $k$ of sampled vertices in $U_i$ is at most $ 2^{2^i}\ln(1/\epsilon)$ in expectation. Order the sampled vertices $v_1,\ldots,v_k$ by value of $d(v_j,u)$, with ties broken arbitrarily. For each vertex $v_j$, if any of $v_1,\ldots,v_{j-1}$ precedes $v_j$ in the random ordering, then Dijkstra's algorithm from $v_j$ skips $u$. It follows that vertex $u$ is added to the priority queue of Dijkstra's algorithm from $v_j$ with probability at most $1/j$. The expected number of times vertex $u$ is added to a priority queue is at most $\frac11+\frac12+\cdots+\frac1k\le\ln k+1$. Taking the expectation over $k$, and using the fact that $\ln k$ is concave, this expectation is at most $\ln( 2^{2^i}\ln(1/\epsilon))+1\le2^i+C\log\log(1/\epsilon)$ for some constant $C>0$.

Next, we bound the expected number of times a vertex $u$ can be added to Dijkstra's algorithm over the entire recursive algorithm. Let $f(m)$ be an upper bound of this expected number on any graph with at most $m$ edges. Given a recursive instance, let $p_i$ be the probability that $u\in U_i$ for each $1\le i\le L$. The probability that $u\in U_i$ and $u\notin U_{i+1}$ is exactly $p_i-p_{i+1}$, where we define $p_{L+1}=0$. (Note that all vertices leave $U$ after iteration $L$ since every vertex is sampled on that iteration.) In particular, $p_i-p_{i+1}$ is the probability that vertex $u$ is cut by a ball on iteration $i$. If $i=1$, then \Cref{thm:2.2} holds with probability at least $1-2m\epsilon$, in which case the recursive instance $G[B^\pm_{G[U_i]}\cap U]$ has at most $\frac34m=1.5m/2^{2^{i-1}}$ edges since all vertices are out-light and in-light. If $i>1$, then \Cref{lem:2} holds on iteration $i$ with probability at least $1-2m\epsilon$, in which case $|E[B^\pm_{G[U_i]}\cap U]|\le\textup{vol}(B^\pm_{G[U_i]}\cap U)/2\le m/2^{2^{i-1}}\le1.5m/2^{2^{i-1}}$. Summing over all $i$, and ignoring the $2m\epsilon$ probabilities of failure for now, we can upper bound $f(m)$ by
\begin{align*}
&\sum_{i=1}^L\left(p_i\cdot(2^i+C\log\log(1/\epsilon))+(p_i-p_{i+1})\cdot f(1.5m/2^{2^{i-1}})\right)
\\={}&\sum_{i=1}^L\left(\big((p_i-p_{i+1})+(p_{i+1}-p_{i+2})+\cdots+(p_L-p_{L+1})\big)\cdot(2^i+C\log\log(1/\epsilon))+(p_i-p_{i+1})\cdot f(1.5m/2^{2^{i-1}})\right)
\\={}&\sum_{i=1}^L(p_i-p_{i+1})\cdot\left(2^1+2^2+\cdots+2^i+i\cdot C\log\log(1/\epsilon)+f(1.5m/2^{2^{i-1}})\right)
\\\le{}&\sum_{i=1}^L(p_i-p_{i+1})\cdot\left(2^{i+1}+i\cdot C\log\log(1/\epsilon)+f(1.5m/2^{2^{i-1}})\right).
\end{align*}
Since the values $(p_i-p_{i+1})$ sum to $1$, the summation is a convex combination of the terms\linebreak $2^{i+1}+i\cdot C\log\log(1/\epsilon)+f(1.5m/2^{2^{i-1}})$, so we can bound the summation by the maximum such term. In conclusion, we have the recursive bound
\[ f(m)\le\max_{i=1}^L\left(2^{i+1}+i\cdot C\log\log(1/\epsilon)+f(1.5m/2^{2^{i-1}})\right) .\]

We now solve the recurrence by $f(m)\le3(C+4)\lg m\log\log(1/\epsilon)$. By induction on $m$, we bound each term in the maximum above by
\begin{align*}
&2^{i+1}+i\cdot C\log\log(1/\epsilon)+f(1.5m/2^{2^{i-1}})
\\\le{}&2^{i+1}+i\cdot C\log\log(1/\epsilon)+3(C+4)\lg(1.5m/2^{2^{i-1}})\log\log(1/\epsilon)
\\={}&2^{i+1}+i\cdot C\log\log(1/\epsilon)+3(C+4)(\lg m+\lg1.5-2^{i-1})\log\log(1/\epsilon)
\\\le{}&2^{i+1}+i\cdot C\log\log(1/\epsilon)+3(C+4)(\lg m-\frac13\cdot2^{i-1})\log\log(1/\epsilon)
\\={}&2^{i+1}+i\cdot C\log\log(1/\epsilon)+3(C+4)\lg m\log\log(1/\epsilon)-(C+4)\cdot2^{i-1}\log\log(1/\epsilon)
\\\le{}&3(C+4)\lg m\log\log(1/\epsilon),
\end{align*}
where the last inequality uses $i\le2^{i-1}$ for all $i\ge1$. It follows that $f(m)\le3(C+4)\lg m\log\log(1/\epsilon)$, completing the induction.

In conclusion, the expected total number of times vertex $u$ is added to a priority queue in Dijkstra's algorithm is $O(\log n\log\log n)$. Using Thorup's priority queue~\cite{thorup2003integer}, each vertex update takes $O(\log\log n)$ time and each edge update takes $O(1)$ time, so the expected total running time is $O((m+n\log\log n)\log n\log\log n)$, as promised.

Finally, we address what happens when the $2m\epsilon$ probability of failure occurs. There are at most $m$ recursive instances, each with $L$ iterations, so the overall failure probability is at most $Lm\cdot O(m\epsilon)\le O(m^3\epsilon)$. In this case, the running time can be trivially bounded by $O(m^3)$. The overall expected running time increase from this failure case is $O(m^3\epsilon\cdot m^3)\le O(m)$ since $\epsilon\le m^{-5}$.
\end{proof}

\subsection{Preprocessing Step}\label{sec:2}

In this section, we present the preprocessing algorithm from \Cref{thm:2.2}, restated below.
\Preprocessing*

Let $C>0$ be a large enough constant. Define vertex subsets $V^+,V^-\subseteq V$ as follows: let $V^\pm$ be all vertices $v\in V$ with $|E[B^\pm_G(v,\Delta/4)]|\ge\frac m{8\ln(2C\ln(1/\epsilon))}$. The algorithm randomly selects one of three cases to follow, and possibly makes a recursive call. For ease of exposition, we first describe the three cases and their guarantees, and then discuss how the algorithm selects which one to follow.

\subsubsection{Case 1}
This case aims to succeed under the assumption that $\textup{vol}(V^+)\ge\frac54m$ and $\textup{vol}(V^-)\ge\frac54m$. Of course, the algorithm cannot tell whether this assumption holds, but stating the assumption upfront may provide useful context. (If the assumption does not hold, the algorithm still cuts balls but may not make any progress.)

The algorithm selects a random vertex $t\in V$ from the distribution that assigns probability $\frac{\deg(v)}{2m}$ to each vertex $v\in V$. The algorithm samples a random radius $r\in[\Delta/4,\Delta/2]$ and cuts the in-ball $B^-_G(t,r)$ and out-ball $B^+_G(t,r)$, adding the cut edges to $S$.

Each strongly connected component in $G-S$ is contained in one of $B^+_G(t,r)\cap B^-_G(t,r)$,\linebreak$B^+_G(t,r)\setminus B^-_G(t,r)$, $B^-_G(t,r)\setminus B^+_G(t,r)$, and $V\setminus(B^+_G(t,r)\cup B^-_G(t,r))$. The algorithm labels all connected components of  $B^+_G(t,r)\cap B^-_G(t,r)$ by property~(\ref{item:b}) of \Cref{thm:2.2}, which is justified below.
\begin{lemma}
$B^+_G(t,r)\cap B^-_G(t,r)$ has weak diameter at most $\Delta$.
\end{lemma}
\begin{proof}
Consider two vertices $u,v\in B^+_G(t,r)\cap B^-_G(t,r)$. Since $u\in B^-_G(t,r)$, we have $d(u,t)\le r$, and since $v\in B^+_G(t,r)$, we have $d(t,v)\le r$. By the triangle inequality, we conclude that $d(u,v)\le d(u,t)+d(t,v)\le r+r\le\Delta$.
\end{proof}
The lemma below shows that if the stated assumption holds, then with constant probability, each remaining component is sufficiently smaller in size.
\begin{lemma}
Assume that $\textup{vol}(V^+)\ge\frac54m$ and $\textup{vol}(V^-)\ge\frac54m$. With probability at least $1/4$, each strongly connected component outside $B^+_G(t,r)\cap B^-_G(t,r)$ has at most $m-\frac m{8\ln(2C\ln(1/\epsilon))}$ edges.
\end{lemma}
\begin{proof}
Under the assumption that $\textup{vol}(V^+)\ge\frac54m$ and $\textup{vol}(V^-)\ge\frac54m$, the sampled vertex $t$ belongs to $V^+$ and $V^-$ each with probability at least $5/8$, so it belongs to their intersection $V^+\cap V^-$ with probability at least $1/4$. In this case, we have $|E[B^\pm_G(t,r)]|\ge|E[B^\pm_G(t,\Delta/8)]|\ge\frac m{8\ln(2C\ln(1/\epsilon))}$. Each component contained in $B^+_G(t,r)\setminus B^-_G(t,r)$, $B^-_G(t,r)\setminus B^+_G(t,r)$, or $V\setminus(B^+_G(t,r)\cup B^-_G(t,r))$ does not contain edges in either $B^+_G(t,r)$ or $B^-_G(t,r)$, so it misses at least $\frac m{8\ln(2C\ln(1/\epsilon))}$ edges, as promised.
\end{proof}
The algorithm recursively solves the largest component if it has more than $\frac34m$ edges, and labels all others by property~(\ref{item:a}) of \Cref{thm:2.2}.

We now bound the probability of cutting each edge. In order for ball $B^+_G(t,r)$ to cut edge $(u,v)$, the radius $r$ must satisfy $d(t,u)\le r<d(t,u)+w(u,v)$, since this holds whenever $u\in B^+_G(t,r)$ and $v\notin B^+_G(t,r)$. There are $\Delta/4$ many choices of $r$, so this happens with probability at most $\frac4\Delta\cdot w(u,v)$. Similarly, the edge $(u,v)$ is cut by ball $B^-_G(s,r)$ with probability at most $\frac4\Delta\cdot w(u,v)$. Overall, each edge $(u,v)$ is cut with probability at most $\frac8\Delta\cdot w(u,v)$.

We summarize the guarantees of Case~1 in the lemma below. Note that the running time is dominated by two calls to Dijkstra's algorithm, which takes $O(m+n\log\log n)$ time using Thorup's priority queue~\cite{thorup2003integer}.
\begin{lemma}\label{lem:2.5}
In Case~1 of the algorithm, each edge $(u,v)$ joins $S$ with probability at most $\frac8\Delta\cdot w(u,v)$. The algorithm correctly labels components by properties~(\ref{item:a})~and~(\ref{item:b}) of \Cref{thm:2.2} and makes a recursive call on the remaining component, if any.
 Under the assumption that $\textup{vol}(V^+)\ge\frac54m$ and $\textup{vol}(V^-)\ge\frac54m$, with constant probability, the recursive component has at most $m-\frac m{8\ln(2C\ln(1/\epsilon))}$ edges. The running time is $O(m+n\log\log n)$.
\end{lemma}

\subsubsection{Case 2}
This case aims to succeed under the assumption that $\textup{vol}(V^+)<\frac54m$. 

The algorithm first samples $8\ln(2C\ln(1/\epsilon))$ vertices with replacement from the distribution that assigns probability $\frac{\deg(v)}{2m}$ to each vertex $v\in V$. The algorithm contracts all sampled vertices that do \emph{not} belong to $V^+$ into a single vertex $t'$ and computes $B^+_{G'}(t',r)$ in the contracted graph $G'$, where $r\in[\Delta/4,\Delta/2]$ is a random radius.

Throughout this case, we study vertices that are out-heavy and in-heavy with respect to different radii and size thresholds. To avoid confusion, we do not use the terms out-heavy and in-heavy unless the radius is specifically $\Delta/8$ and the size threshold is specifically $\frac34m$ as defined.
\begin{lemma}\label{lem:2.6}
Assume that $\textup{vol}(V^+)<\frac54m$. For each vertex $v\in V$ with $|E[B^-_G(v,\Delta/4)]|>\frac34m$, we have $v\notin B^+_{G'}(t',r)$ with probability at most $\frac1{2C\ln(1/\epsilon)}$.
\end{lemma}
\begin{proof}
Consider a vertex $v\in V$ with $|E[B^-_G(v,\Delta/4)]|>\frac34m$, which means $\textup{vol}(B^-_G(v,\Delta/4))>\frac32m$. Together with the assumption $\textup{vol}(V^+)<\frac54m$, we have $\textup{vol}(B^-_G(v,\Delta/4)\setminus V^+)>\frac14m$. If the algorithm samples a vertex in $B^-_G(v,\Delta/4)\setminus V^+$, then that vertex is contracted into $t'$, and we conclude that $v\in B^+_{G'}(t',\Delta/4)\subseteq B^+_{G'}(t',r)$. Conversely, the probability that $v\notin B^+_{G'}(t',r)$ is at most the probability of sampling no vertex in $B^-_G(v,\Delta/4)\setminus V^+$. Due to negative correlation, this probability is at most the product of the probabilities of not sampling each vertex in $B^-_G(v,\Delta/4)\setminus V^+$. This product is bounded by
\begin{align*}
&\prod_{v\in B^-_G(v,\Delta/4)\setminus V^+}\left(1-\frac{\deg(v)}{2m}\right)^{8\ln(2C\ln(1/\epsilon))}
\\\le{}&\prod_{v\in B^-_G(v,\Delta/4)\setminus V^+}\exp\left(-8\ln(2C\ln(1/\epsilon))\cdot\frac{\deg(v)}{2m}\right)
\\={}&\exp\left(-8\ln(2C\ln(1/\epsilon))\cdot\frac{\textup{vol}(B^-_G(v,\Delta/4)\setminus V^+)}{2m}\right)
\\\le{}&\exp\bigg(-8\ln(2C\ln(1/\epsilon))\cdot\frac{m/4}{2m}\bigg)
\\={}&\exp(-\ln(2C\ln(1/\epsilon))
\\={}&\frac1{2C\ln(1/\epsilon)},
\end{align*}
as promised.
\end{proof}

The algorithm emulates cutting the ball $B^+_{G'}(t',r)$ in $G$, which means adding all cut edges to $S$ (after undoing the contraction) and removing all (original) vertices in $B^+_{G'}(t',r)$, namely the union of $B^+_G(t,r)$ over all vertices $t\in V$ that were contracted to $t'$. Since each contracted vertex $t$ does not belong to $V^+$, we have $\textup{vol}(B^+_G(t,r))<\frac m{8\ln(2C\ln(1/\epsilon))}$, so the total volume of removed vertices is at most $8\ln(2C\ln(1/\epsilon))\cdot\frac m{8\ln(2C\ln(1/\epsilon))}=m$. In other words, at most $m/2$ edges are removed. The algorithm labels all strongly connected components of removed vertices by property~(\ref{item:a}) of \Cref{thm:2.2}.

Let $U$ be the remaining vertices. The algorithm executes the following, which is essentially steps~(\ref{step:a}) to~(\ref{step:c}) of \Cref{algo:1} for iteration $i=1$, except that only in-balls are cut (and the parameters are slightly different).

\begin{algo}\label{algo:6} \
 \begin{enumerate}
 \item[1.] Initialize $U$ as the remaining vertices. Let $U_1$ be a snapshot of $U$ at this time.
 \item[(3a)] Sample each vertex $v\in U$ with probability $C\ln(1/\epsilon)\cdot\frac{\deg(v)}{2m}$. Let $S_1$ be the random sample.
 \item[(3b)] Sample a random radius $r_1\in[\frac\Delta6,\frac\Delta4]$.
 \item[(3c)] For each vertex $v\in S_1$ in a random order:
  \begin{enumerate}
  \item[i.] Cut the in-ball $B^-_{G[U_1]}(v,r_1)$ from the remaining graph $G[U]$. That is, add the edges of $G[U]$ cut by $B^-_{G[U_1]}(v,r_1)$ to $S$, and then set $U\gets U\setminus B^-_{G[U_1]}(v,r_1)$.
  \end{enumerate}
 \end{enumerate}
\end{algo}

The lemma below implies that all out-heavy vertices in $G[U_1]$ are eliminated with good probability. In fact, the statement is even stronger, with a radius of $\Delta/6$ instead of $\Delta/8$.
\begin{lemma}\label{lem:2.8}
With probability at least $1-m\epsilon$, the final set $U$ has no vertices $v\in V$ satisfying $|E[B^+_{G[U_1]}(v,\Delta/6)]|>\frac34m$.
\end{lemma}
\begin{proof}
Consider a vertex $u\in V$ satisfying $|E[B^+_{G[U_1]}(u,\Delta/6)]|>\frac34m$. By a similar calculation to the proof of \Cref{lem:2}, with probability at least $1-\epsilon$, at least one vertex $v\in B^+_{G[U_1]}(u,\Delta/6)$ is sampled. In this case, since $r_1\ge\Delta/6$, we also have $v\in B^+_{G[U_1]}(u,r_1)$, or equivalently $u\in B^-_{G[U_1]}(v,r_1)$, so vertex $u$ is removed from $U$ when cutting out $B^-_{G[U_1]}(v,r_1)$ (if not earlier). Taking a union bound over all out-heavy vertices, we conclude that with probability at least $1-m\epsilon$, the final set $U$ has no out-heavy vertices.
\end{proof}

\begin{lemma}\label{lem:2.9}
Assume that $\textup{vol}(V^+)<\frac54m$. With probability at least $1/2$, each in-ball $B^-_{G[U_1]}(v,r_1)$ cut by the algorithm satisfies $|E[B^+_{G[U_1]}(v,r_1)]|\le\frac34m$.
\end{lemma}
\begin{proof}
By \Cref{lem:2.6}, for each vertex $v\in V$ with $|E[B^-_G(v,\Delta/4)]|>\frac34m$, we have $v\notin B^+_{G'}(t',r)$ with probability at most $\frac1{2C\ln(1/\epsilon)}$. Conditioned on this event for vertex $v\in V$, the algorithm samples that vertex with probability $C\ln(1/\epsilon)\cdot\frac{\deg(v)}{2m}$. Therefore, the overall probability of cutting in-ball $B^+_{G[U_1]}(v,r_1)$ is at most $\frac1{2C\ln(1/\epsilon)}\cdot C\ln(1/\epsilon)\cdot\frac{\deg(v)}{2m}=\frac{\deg(v)}{4m}$. Summing over all $v\in V$ with $|E[B^-_G(v,\Delta/4)]|>\frac34m$, the expected number of such in-balls is at most $\sum_{v\in V}\frac{\deg(v)}{4m}=1/2$. It follows that with probability at least $1/2$, no such in-ball is cut, which means all in-balls $B^+_G(v,r_1)$ satisfy $|E[B^+_{G[U_1]}(v,r_1)]|\le|E[B^+_G(v,\Delta/4)]|\le\frac34m$.
\end{proof}

If any in-ball $B^-_{G[U_1]}(v,r_1)$ cut by the algorithm satisfies $|E[B^+_{G[U_1]}(v,r_1)]|>\frac34m$, then the algorithm gives up and recursively solves the original $G[U_1]$ (ignoring the output of \Cref{algo:6}). By \Cref{lem:2.9}, if the assumption $\textup{vol}(V^+)<\frac54m$ holds, then the algorithm continues with probability at least $1/2$. In this case, the algorithm labels all strongly connected components of removed vertices by property~(\ref{item:a}) of \Cref{thm:2.2}. Then, the algorithm essentially repeats \Cref{algo:6} on the remaining vertices except that out-balls are now cut, and the radius range is now $[\frac\Delta8,\frac\Delta6]$.

\begin{algo}\label{algo:9} \
 \begin{enumerate}
 \item[1.] Initialize $U$ as the remaining vertices. Let $U_2$ be a snapshot of $U$ at this time.
 \item[(3a)] Sample each vertex $v\in U_2$ with probability $C\ln(1/\epsilon)\cdot\frac{\deg(v)}{2m}$. Let $S_2$ be the random sample.
 \item[(3b)] Sample a random radius $r_2\in[\frac\Delta8,\frac\Delta6]$.
 \item[(3c)] For each vertex $v\in S_2$ in a random order:
  \begin{enumerate}
  \item[i.] Cut the out-ball $B^+_{G[U_2]}(v,r_2)$. That is, add the edges of $G[U]$ cut by $B^+_{G[U_2]}(v,r_2)$ to $S$, and then set $U\gets U\setminus B^+_{G[U_2]}(v,r_2)$.
  \end{enumerate}
 \end{enumerate}
\end{algo}

\begin{lemma}\label{lem:2.10}
With probability at least $1-2m\epsilon$, all remaining vertices in $U$ are both out-light and in-light in $G[U]$, and each out-ball $B^+_{G[U_2]}(v,r_2)$ cut by the algorithm satisfies $|E[B^+_{G[U_2]}(v,r_2)]|\le\frac34m$.
\end{lemma}
\begin{proof}
By \Cref{lem:2.8}, all vertices in $U_2$ satisfy $|E[B^+_{G[U_1]}(v,\Delta/6)]|\le\frac34m$ with probability at least $1-m\epsilon$, so each vertex in $U\subseteq U_2$ is out-light in $G[U_2]$, and hence also in $G[U]$. Since $r_2\le\Delta/6$, each out-ball $B^+_{G[U_2]}(v,r_2)$ cut by the algorithm satisfies $|E[B^+_{G[U]}(u,r_2)]|\le|E[B^+_{G[U_1]}(v,r_2)]|\le\frac34m$.

It remains to show that all remaining vertices in $U$ are in-light in $G[U]$. Consider a vertex $u\in V$ with $|E[B^-_{G[U_2]}(u,\Delta/8)]|>\frac34m$. By a similar calculation to the proof of \Cref{lem:2}, with probability at least $1-\epsilon$, at least one vertex $v\in B^-_{G[U_2]}(u,\Delta/8)$ is sampled. In this case, since $r_2\ge\Delta/8$, we also have $v\in B^-_{G[U_2]}(u,r_2)$, or equivalently $u\in B^+_{G[U_2]}(v,r_2)$, so vertex $u$ is removed from $U$ when cutting out $B^+_{G[U_2]}(v,r_2)$ (if not earlier). Taking a union bound over all vertices, we conclude that with probability at least $1-m\epsilon$, all vertices $v$ in the final set $U$ satisfy $|E[B^-_{G[U_2]}(u,\Delta/8)]|\le\frac34m$, so they are in-light in $G[U_2]$, and hence also in $G[U]$.
\end{proof}
The algorithm labels all strongly connected components in the final set $U$ by property~(\ref{item:c}) of \Cref{thm:2.2}, and all components of removed vertices by property~(\ref{item:a}), which is correct with high probability by~\Cref{lem:2.10}. No recursive call is made.  

We now bound the probability of cutting each edge. In Algorithms~\ref{algo:6} and~\ref{algo:9}, each edge is cut with probability $O(\frac{\log\log(1/\epsilon)}{\Delta})\cdot w(u,v)$, which can be analyzed identically to the proof of \Cref{lem:4} (for iteration $i=1$). The initial out-ball $B^+_{G'}(t',r)$ cuts each edge with probability at most $\frac4\Delta\cdot w(u,v)$.

Finally, we bound the total running time. In the initial sample of $8\ln(2C\ln(1/\epsilon))$ vertices, the algorithm determines which vertices belong to $V^+$, which is a Dijkstra computation for each vertex and runs in $O((m+n\log\log n)\log\log(1/\epsilon))$ total time. The running times of Algorithms\ref{algo:6} and~\ref{algo:9} can be analyzed identically to \Cref{lem:6}, using a speedup of Dijkstra's algorithm and the fact that the ordering is random, for a total time of $O((m+n\log\log n)\log\log(1/\epsilon))$.

We summarize the guarantees of Case~2 in the lemma below.
\begin{lemma}\label{lem:2.12}
In Case~2 of the algorithm, each edge $(u,v)$ joins $S$ with probability $O(\frac{\log\log(1/\epsilon)}{\Delta})\cdot w(u,v)$. There are two outcomes of the algorithm. In the first outcome, the algorithm correctly labels components by property~(\ref{item:a}) of \Cref{thm:2.2} and makes a recursive call on the remaining component, if any. In the second outcome, the algorithm labels components by properties~(\ref{item:a})~and~(\ref{item:c}) which is correct with probability at least $1-2m\epsilon$, and makes no recursive call. Under the assumption that $\textup{vol}(V^+)<\frac54m$, the second outcome occurs with constant probability. The running time is $O((m+n\log\log n)\log\log(1/\epsilon))$.
\end{lemma}

\subsubsection{Case 3}
This case is symmetric to Case~2 and aims to succeed under the assumption that $\textup{vol}(V^-)<\frac54m$. The algorithm is essentially the same, with out-balls replaced by in-balls and vice versa. The guarantees are identical to \Cref{lem:2.12}, except that the assumption is now $\textup{vol}(V^-)<\frac54m$.

\subsubsection{Combining all three cases}
The final algorithm executes Case~2 and Case~3 each with probability $\frac1{2\ln\ln(1/\epsilon)}$, and Case~1 with the remaining probability $1-\frac1{\ln\ln(1/\epsilon)}$. If there is a strongly connected component of $G-S$ that does not satisfy any of the conditions in property~(\ref{item:2}) of \Cref{thm:2.2}, then the algorithm recursively processes that component. Such a component must have more than $\frac34m$ edges, so there is at most one such component. By design, the algorithm finishes with a set $S$ that satisfies property~(\ref{item:2}).

It remains to bound the probability that each edge joins $S$, as well as the running time. We measure progress by the size of the recursive component, since the algorithm stops once the component falls below $\frac34m$ edges. (If there is no recursive component, then the algorithm immediately finishes.)

First, suppose that $\textup{vol}(V^+)\ge\frac54m$ and $\textup{vol}(V^-)\ge\frac54m$, which is the assumption for Case~1. There is a constant probability of selecting Case~1, and conditioned on that, the recursive component (if any) has $\Omega(\frac m{\log\log(1/\epsilon)})$ fewer edges with constant probability by \Cref{lem:2.5}. It follows that the recursive component has $\Omega(\frac m{\log\log(1/\epsilon)})$ fewer edges in expectation.

Next, suppose that $\textup{vol}(V^+)<\frac54m$. The algorithm selects Case~2 with probability $\frac1{2\ln\ln(1/\epsilon)}$, and conditioned on that, the algorithm finishes with constant probability by \Cref{lem:2.12}. If we interpret a finished algorithm as having a recursive component of no edges, then the recursive component also has $\Omega(\frac m{\log\log(1/\epsilon)})$ fewer edges in expectation.

Finally, the case $\textup{vol}(V^+)<\frac54m$ is symmetric. Overall, each case decreases the number of edges by $\Omega(\frac m{\log\log(1/\epsilon)})$ in expectation, so the expected number of recursive calls is $O(\log\log(1/\epsilon))$.\footnote{This can be made rigorous by treating the process as a (super-)martingale and applying the optional stopping theorem. More precisely, define the super-martingale as the number of edges remaining plus $t\cdot\Omega(\frac m{\log\log1/\epsilon})$, where $t$ is the number of recursive calls so far. If $T$ is the total number of rounds, then by the optional stopping theorem, $\mathbb E[T\cdot \Omega(\frac m{\log\log1/\epsilon})]\le m$, as promised.} By \Cref{lem:2.5,lem:2.12}, the probability that an edge $(u,v)$ joins $S$ is $O(\frac1\Delta)\cdot w(u,v)$ in Case~1 and $O(\frac{\log\log(1/\epsilon)}{\Delta})\cdot w(u,v)$ in Cases~2~and~3. Together with the probabilities of selecting each case, the overall probability that edge $(u,v)$ joins $S$ on a given call of the algorithm is $O(\frac1\Delta)\cdot w(u,v)$. Since the expected number of calls is $O(\log\log(1/\epsilon))$, the total probability that edge $(u,v)$ joins $S$ is $O(\frac{\log\log(1/\epsilon)}{\Delta})\cdot w(u,v)$, as promised.

It remains to bound the expected running time. By \Cref{lem:2.5,lem:2.12}, the algorithm runs in $O(m+n\log\log n)$ time in Case~1 and $O((m+n\log\log n)\log\log(1/\epsilon))$ time in Cases~2~and~3. Together with the probabilities of selecting each case, the expected running time is $O(m+n\log\log n)$ on a given call of the algorithm. Since the expected number of calls is $O(\log\log(1/\epsilon))$, the total expected running time is $O((m+n\log\log n)\log\log(1/\epsilon))$, as promised.

\subsection{Decomposition Interface}\label{sec:decomposition-interface}

For the rest of the paper, we focus on negative-weight SSSP. We start with the directed LDD interface used by the negative-weight SSSP algorithm. This interface is essentially \Cref{algo:1} except that instead of iterating over $i\in\{1,2,\ldots,L\}$ in a single recursive instance, the algorithm tracks $i$ as a recursive parameter and creates a separate recursive instance for each iteration $i$. Since this recursive instance may have fewer edges than the initial instance of \Cref{algo:1}, the algorithm also needs to track the original number $m_0$ of edges. The other difference is that the algorithm still recursively solves components with weak diameter $\Delta$, namely the components labeled (\ref{item:b}), with $\Delta/2$ as the new diameter parameter. For simplicity, we do not specify the base cases, deferring that to the negative weight SSSP algorithm.

\begin{lemma}\label{lem:decompose}
There is a recursive algorithm $\textsc{Decompose}(H,\Delta,i,m_0)$, where $m_0\ge|E(H)|$, that cuts a random subset $S\subseteq E$ of edges and recursively decomposes strongly connected components (SCCs) of $H-S$ as follows:
 \begin{enumerate}
 \item If $i=0$, then there are three possibilities for a SCC $H'$ of $H-S$:
  \begin{enumerate}
  \item The SCC $H'$ has at most $\frac34m_0$ edges. In this case, the algorithm recursively calls\linebreak $\textsc{Decompose}(H',\Delta,0,m_{H'})$ on this component.\label{item:a2}
  \item The SCC $H'$ is certified to have weak diameter at most $\Delta/2$ in $H$. In this case, the algorithm recursively calls $\textsc{Decompose}(H',\Delta/2,0,m_{H'})$ on this component.\label{item:b2}
  \item The SCC $H'$ satisfies neither case above. Let $U\subseteq V$ be the union of all such SCCs. The algorithm recursively calls $\textsc{Decompose}(H[U],\Delta,1,m_0)$.
  \end{enumerate}
 Each edge $(u,v)$ joins $S$ with probability at most $O(\frac{\log\log(1/\epsilon)}\Delta)\cdot w(u,v)$. The algorithm takes\linebreak $O((m+n\log\log n)\log\log(1/\epsilon))$ expected time and succeeds with probability at least $1-2m_0\epsilon$.
 \item If $i>0$, then there are two possibilities:
  \begin{enumerate}
  \item The component $H'$ has at most $1.5m_0/2^{2^{i-1}}$ edges. In this case, the algorithm recursively calls $\textsc{Decompose}(H',\Delta,0,m_{H'})$ on this component.
  \item Let $U\subseteq V$ be the union of all remaining SCCs. The algorithm recursively calls\linebreak $\textsc{Decompose}(H[U],\Delta,i+1,m_0)$. When $i=\lceil\lg\lg m_0\rceil$, there are no remaining components, so there is no recursive call.
  \end{enumerate}
 Each edge $(u,v)$ joins $S$ with probability at most $O(\frac{2^i\log\log(1/\epsilon)}\Delta)\cdot w(u,v)$. The algorithm takes\linebreak $O((m+n\log\log n)\cdot(2^i+\log\log(1/\epsilon)))$ expected time and succeeds with probability at least $1-2m_0\epsilon$.
  \end{enumerate}
Moreover, suppose that the original instance is $(G,\Delta_0,0,|E(G)|)$. In each recursive instance $(H,\Delta,i,m_0)$ where $\Delta<\Delta_0$, the subgraph $H$ has weak diameter at most $\Delta$ in $G$.
\end{lemma}
\begin{proof}
We simulate \Cref{algo:1} with diameter parameter $\Delta/2$ by breaking the algorithm into more recursive steps, starting with $i=0$ as the recursive parameter. On the preprocessing step, the algorithm calls \Cref{thm:2.2} and makes appropriate recursive calls to components labeled (\ref{item:a}) and (\ref{item:b}). The desired properties are guaranteed by \Cref{thm:2.2}. Then, instead of continuing \Cref{algo:1}, the algorithm makes a recursive call on $H[U]$ with parameters $i=1$ and $m_0=m$.

On a recursive instance with $i>0$, the algorithm skips the preprocessing step and simulates the loop at step~\ref{item:for} for that value of $i$, except that the sampling probability is $\min\{1,\,C\cdot 2^{2^i}\ln(1/\epsilon)\cdot\frac{\deg(v)}{2m_0}\}$ for large enough $C>0$ to guarantee that \Cref{lem:2} holds with high probability. The algorithm makes a recursive call on each cut ball with parameter $i=0$, and makes a call with parameter $i+1$ on the remaining vertices. Note that when $i=\lceil\lg\lg m_0\rceil$, all vertices are sampled, so there are no remaining vertices. When $i>1$, each cut ball has volume at most $2m_0/2^{2^{i-1}}$, so each recursive instance has at most $m_0/2^{2^{i-1}}$ edges. When $i=1$, each cut ball has at most $\frac34m_0$ edges, so in both cases, the number of edges is at most $1.5m_0/2^{2^{i-1}}$. Then, instead of continuing onto iteration $i+1$, the algorithm makes a recursive call on $H[U]$ with parameter $i+1$ and the same value of $m_0$. By \Cref{lem:4}, each edge is cut with probability $O(\frac{2^i\log\log(1/\epsilon)}{\Delta})\cdot w(u,v)$. Finally, by inspecting the proof of \Cref{lem:6}, the expected number of times each vertex is added to the priority queue is at most $2^i+O(\log\log(1/\epsilon))$, so the algorithm runs in $O((m+n\log\log n)\cdot(2^i+\log\log(1/\epsilon)))$ time using Thorup's priority queue~\cite{thorup2003integer}.

Finally, suppose that the original instance is $(G,\Delta_G,0,m_G)$. Consider a recursive instance $(H,\Delta,i,m_0)$. If it is a recursive call from step~(\ref{item:b2}) of the parent instance, then $H$ has weak diameter at most $\Delta$ in the parent graph. Since the parent graph is a subgraph of $G$, we conclude that $H$ has weak diameter at most $\Delta$ in $G$. Otherwise, if recursive instance $(H,\Delta,i,m_0)$ is a result of a parent instance with the same parameter $\Delta<\Delta_G$, then by induction (from the root to the leaves), the parent graph has weak diameter at most $\Delta$ in $G$. Since $H$ is a subgraph of the parent graph, it also has weak diameter at most $\Delta$ in $G$.
\end{proof}

\section{Negative-Weight Single Source Shortest Paths Algorithm}\label{sec:sssp-main}

For the rest of the paper, we focus on negative-weight single source shortest paths, proving our main theorem below.
\begin{theorem}\label{thm:main}
There is a randomized algorithm for negative-weight single source shortest paths that runs in $O((m+n\log\log n)\log(nW)\log n\log\log n)$ time and succeeds with high probability, where $W$ is the maximum absolute value of a negative weight edge. If a negative weight cycle exists, then the algorithm instead outputs a negative weight cycle.
\end{theorem}

At a high level, the algorithm follows the framework of~\cite{bringmann2023negative} replacing their recursive parameter $\kappa(G)$ with weak diameter, as previously done by~\cite{fischer2025simple}. The algorithm of~\cite{bringmann2023negative} loses a logarithmic factor from each recursive decomposition step, as well as another from the recursion depth. Instead, we use the specialized decomposition procedure from \Cref{lem:decompose}, which merges both logarithmic factors into a single $O(\log n\log\log n)$, obtaining the improvement.

In an orthogonal direction, we also remove the technical noisy binary search step from~\cite{bringmann2023negative}, which is used to recover a negative weight cycle if one exists. Instead, we show that a simple augmented Bellman-Ford/Dijkstra algorithm can successfully recover one.

\subsection{Additional Preliminaries}\label{sec:sssp-prelims}

A \emph{price function} $\phi:V\to\mathbb R$ is a function on the vertices used to reweight the graph: the reweighted graph $G_\phi$ satisfies $w_{G_\phi}(u,v)=w_G(u,v)+\phi(u)-\phi(v)$ for all edges $(u,v)\in E$. A result of Johnson~\cite{johnson1977efficient} reduces the Negative-Weight SSSP problem to computing a potential function $\phi$ such that all edge weights of $G_\phi$ are non-negative, after which Dijkstra's algorithm can be run on $G_\phi$ to recover the single-source shortest paths in $G$.

The iterative \emph{scaling technique} of~\cite{bernstein2025negative} does not compute such a potential function in one go. Instead, if the graph $G$ has all edge weights at least $-W$, then the goal is to either compute a potential function $\phi$ such that $G_\phi$ has all edge weights at least $-W/2$, or output a negative-weight cycle; we call this problem \textsc{Scale}.~\cite{bernstein2025negative} show that Negative-Weight SSSP reduces to $O(\log(nW))$ iterations of \textsc{Scale} when the original graph has all edge weights at least $-W$. Moreover,~\cite{bringmann2023negative} show that if each iteration of \textsc{Scale} runs in $T$ time \emph{in expectation}, then Negative-Weight SSSP can be computed in $O(T\log(nW))$ time \emph{with high probability} (and with Las Vegas guarantee) by restarting any runs of \textsc{Scale} that exceed $2T$ time and applying a Chernoff bound. Hence, to obtain Theorem~\ref{thm:main}, we only focus on the task of solving \textsc{Scale} in $O((m+n\log\log n)\log n\log\log n)$ expected time. It even suffices to solve \textsc{Scale} in this time with high probability, since the solution can be easily checked for correctness and \textsc{Scale} can be restarted if incorrect.

We need two main subroutines from~\cite{bernstein2025negative,bringmann2023negative}. The first subroutine is a Bellman-Ford/Dijkstra hybrid algorithm for Negative-Weight SSSP, also introduced in~\cite{bernstein2025negative} and refined in~\cite{bringmann2023negative}. For convenience, we do not mention a source vertex in the input to \textsc{BellmanFordDijkstra}. Instead, the goal is to compute, for each vertex $v\in V$, the shortest path ending at $v$ (and starting from anywhere). Note that the shortest path has weight at most $0$ since the empty path at $v$ of weight $0$ is a possible choice. To frame this problem as SSSP, simply add a source vertex $s$ with zero-weight edges to each vertex in $V$, and call SSSP on the new graph $G_s$ with source $s$. For further convenience, we also integrate a potential function $\phi:V\to\mathbb R$ directly into \textsc{BellmanFordDijkstra}, which runs the Bellman-Ford/Dijkstra hybrid on $(G_s)_\phi$ instead (with the extension $\phi(s)=0$), and quickly finds single-source shortest paths with few negative-weight edges in $(G_s)_\phi$. These single-source shortest paths in $(G_s)_\phi$ correspond exactly to single-source shortest paths in $G_s$, which in turn correspond to shortest paths in $G$ ending at each vertex.

In the version below, we also include the algorithm's guarantees after each iteration $i$ to allow for early termination if a negative-weight cycle is detected. In addition, we also maintain the weights of our current shortest paths under a separate, auxiliary weight function $w'$. This can be implemented by recording the distance of each current path $P$ as a tuple $(w_G(P),w'(P))$, but only using the first coordinate for comparisons. For completeness, we include this proof in the appendix.

\begin{restatable}[Lemmas 25 and 62 of~\cite{bringmann2023negative}]{lemma}{BFD}\label{lem:BFD}
Let $G=(V,E,w_G)$ be a directed graph and $\phi$ be a potential function. For each vertex $v\in V$, let $\eta_{G,\phi}(v)$ be the minimum number of edges of negative weight in $G_\phi$ among all shortest paths in $G$ ending at $v$. There is an algorithm \textsc{BellmanFordDijkstra}$(G,\phi)$ that computes a shortest path ending at each vertex in time\linebreak $O(\sum_v(\deg(v)+\log\log n)\cdot\eta_{G,\phi}(v))$.

The first $i$ iterations of the algorithm runs in $O((m+n\log\log n)\cdot i)$ time and computes, for each vertex $v\in V$, the minimum weight $d_i(v)$ of a path $P_i(v)$ from $s$ to $v$ among those that contains less than $i$ negative-weight edges. Moreover, given input auxiliary weights $w'$ on the edges (independent of the weights $w_G$ that determine shortest paths), the algorithm can also compute the weight $d_i'(v)$ under $w'$ of some such path $P_i(v)$.
\end{restatable}

The second subroutine is given a graph with its strongly connected components and a potential function, and fixes the DAG edges outside the SCCs while keeping the edge weights inside the SCCs non-negative.
\begin{lemma}[Lemma 3.2 of~\cite{bernstein2025negative}]\label{lem:fix-dag}
Consider a graph $G$ with strongly connected components $C_1,\ldots,C_\ell$ and a potential function $\phi$, and suppose that all edges inside the SCCs have non-negative weight in $G_\phi$. There is an $O(n+m)$ time algorithm that adjusts $\phi$ so that all edge weights in $G_\phi$ are now non-negative.
\end{lemma}

We say that a vertex subset $U\subseteq V$ has (weak) diameter at most $d$ in a graph $G=(V,E)$ if for any two vertices $u,v\in U$, there is a path in $G$ from $u$ to $v$ of weight at most $d$.
Let $G_{\ge 0}$ be $G$ with negative weight edges replaced by edges of weight 0.
That is, $w_{G_{\ge 0}}(e) = \max\{w_G(e),0\}$.
For a given graph $H$, let \( n_H = |V(H)| \) and \( m_H = |E(H)| \) denote the number of vertices and edges in $H$.

\subsection{The Scale Algorithm}

Recall the specifications of the \textsc{Scale} problem: given an input graph \( G = (V, E, w_G) \) with edge weights \( w_G(e) \geq -W \) for all \( e \in E \), either return a reweighted graph \( G_\phi \) with weights at least \( -W/2 \) or output a negative-weight cycle in \( G \). We present an algorithm that runs in \( O\left( (m + \log \log n) \log n \log \log n \right) \) expected time and succeeds with high probability, which is sufficient to prove Theorem~\ref{thm:main}. 

The algorithm operates in two main phases: (1) recursively decomposing a scaled version of the input graph, and (2) iteratively computing distances on the decomposition.

\subsubsection{Phase 1: Recursive Decomposition}

Let \( G = (V, E, w_G) \) be the input graph. We define a scaled graph \( G' = (V, E, w_{G'}) \) by increasing the weight of every edge by $W/2$, i.e., $w_{G'}(e)=w_G(e)+W/2$ for all edges $e\in E$.

The recursive decomposition calls \textsc{Decompose}$(G'_{\ge0},nW/2,0,m)$ from \Cref{lem:decompose}, i.e., with initial diameter $nW/2$, and stops the recursion if either a single vertex remains, or the diameter falls below $W/2$. The recursive process creates a decomposition tree, where each recursive instance $(H,d,i,m_0)$ is represented by a node $(H,d)$ with tag $(i,m_0)$, whose children are the recursive calls made by this instance.

Once the decomposition is complete, we check each leaf node for the presence of an edge $(u,v)$ that is negative in $G'$. If such an edge exists, then we claim there must be a negative-weight cycle in $G$. To find a negative-weight cycle, we run Dijkstra's algorithm from \( v \) in \( G'_{\ge0} \) (where negative-weight edges in \( G' \) are replaced by edges of weight 0) to find a path to \( u \) of weight at most $W/2$ in $G'_{\ge 0}$; we will show that such a path must exist. Note that this path also has weight at most $W/2$ in $G$, where edge weights are only smaller. The algorithm then outputs the cycle formed by concatenating this path with the edge \( (u,v) \), which is a negative-weight cycle since the edge $(u,v)$ has negative weight in $G'$ and hence weight less than $-W/2$ in $G$.

At this point, if the algorithm does not terminate early with a negative-weight cycle, then all edges in leaf nodes are non-negative in $G'$.

\subsubsection{Phase 2: Iterative Distance Computation}

In this phase, we iteratively compute distance estimates on the decomposition tree, starting from the leaves and moving up to the root. Initially, we set \( \phi(v) = 0 \) for all \( v \in V \). For each node \( (H, d) \) in the decomposition tree, we skip further computation if \( (H, d) \) is a leaf, as all edges in \( H \) are non-negative. For a non-leaf node $(H,d)$, assume that we have already computed a $\phi$ such that all SCCs of \( (H \setminus S_{(H, d)})_\phi \) have non-negative weights. Using Lemma~\ref{lem:fix-dag}, we adjust $\phi$ so that all edges of \( (H \setminus S_{(H, d)})_\phi \) now have non-negative weights. 
We then run \textsc{BellmanFordDijkstra} on $H$ with potential function $\phi$, and use $w_{G'_{\ge 0}}$ as an auxiliary function to help detect a negative-weight cycle.

If \textsc{BellmanFordDijkstra} records a path with auxiliary weight more than $d$, then we terminate \textsc{BellmanFordDijkstra} early and recover such a path $P$ with endpoints $u$ and $v$. 
We then run Dijkstra's algorithm from \( v \) in \( G'_{\ge0} \) to find a shortest path \( P' \) back to \( u \). We show that concatenating \( P \) and \( P' \) produces a negative-weight cycle in $G$. 

After processing all nodes at the current level, we update \( \phi \) and proceed to the next level. After processing all the levels, we guarantee that $G'_\phi$ has non-negative weight edges everywhere.


\begin{algorithm}[H]
\caption{\textsc{Scale} Algorithm}
\begin{algorithmic}[1]
\Function{Scale}{$G = (V, E, w_G), W$}
    \Comment{Input: Graph $G$ with $w_G(e) \geq -W$}
    \ForAll{edges $e \in E$}
        \State $w_{G'}(e) \gets w_G(e) + W/2$
    \EndFor
    \State $G' \gets (V, E, w_{G'})$

    \State \textbf{// Phase 1: Recursive Decomposition}
    \State $d_0 \gets nW/2$
    \State Call the recursive \Call{Decompose}{$G'_{\ge0},d_0,0,m$} with $\epsilon=m^{-5}$. stopping the recursion when $d\le W/2$

    \State \textbf{// Phase 2: Iterative Distance Computation}
    \State $\phi \gets \ $\Call{RecursiveDistanceComputation}{$G',d_0$}

    \State \Return the final potential function $\phi$
\EndFunction
\end{algorithmic}
\end{algorithm}

\begin{algorithm}[H]
\caption{Phase 2: \textsc{RecursiveDistanceComputation} of $G'$ with decomposition tree $\mathcal T$ on $G'_{\ge0}$}
\begin{algorithmic}[1]
\Function{RecursiveDistanceComputation}{$H,d$}
    \If{$(H, d)$ is a leaf}
        \If{\textbf{exists} edge $(u,v) \in E(H)$ with $w_{G'}(u,v) < 0$}
            \State Run Dijkstra's algorithm from $v$ in $G'_{\ge0}$
            \State Let $P$ be the path from $v$ to $u$ in $G'_{\ge0}$
            \State \Return negative-weight cycle formed by $P + (u,v)$
        \Else
            \State \Return$\phi(v)\gets0$ for all $v\in V(H)$ \Comment{All edges are non-negative}
        \EndIf
    \Else
        \ForAll{child nodes $(H',d')$ in $\mathcal T$}
            \State Make a recursive call $\Call{RecursiveDistanceComputation}{H',d'}$
            \If{a negative-weight cycle $C$ is returned}
                \State \Return negative-weight cycle $C$
            \Else
                \State Let $\phi_{(H',d')}$ be the recursive output
            \EndIf
        \EndFor
        \State Merge all $\phi_{(H',d')}$ into a single $\phi_0$ \Comment{$\phi(v)=\phi_{(H',d')}(v)$ for the component $H'$ containing $v$}
        \State Adjust $\phi_0$ so DAG edges are non-negative \Comment{Lemma~\ref{lem:fix-dag}}
        \State Run \textsc{BellmanFordDijkstra}$(H,w_{G'}|_H,\phi_0)$ with auxiliary weight $w_{H_{\ge 0}}$, performing:
        \ForAll{iterations $j$ during \textsc{BellmanFordDijkstra}}
             \State Compute values $d_j(v)$ and $d'_j(v)$ for all $v\in V$
             \If{$d'_j(v) > d$ for some vertex $v$}
                        \State Recover the path $P$ from $u$ to $v$ with auxiliary weight exceeding $d$
                        \State Run Dijkstra's algorithm from $v$ in $G'_{\ge0}$
                        \State Let $P'$ be the path from $v$ to $u$ in $G'_{\ge0}$
                        \State \Return negative-weight cycle formed by $P + P'$
                    \EndIf
                \EndFor
                \State Return $\phi(v)\gets$ computed distance for all $v \in V(H)$
            \EndIf
\EndFunction
\end{algorithmic}
\end{algorithm}

\subsection{Analysis}

Our ultimate goal is to prove the following theorem.

\begin{restatable}{theorem}{Master}\label{thm:master}
    There is an algorithm \textsc{Scale} such that, for any input graph \( G = (V, E, w_G) \) with edge weights \( w_G(e) \geq -W \) for all \( e \in E \), \textsc{Scale}$(G)$ either returns a reweighted graph \( G_\phi \) with edge weights at least \( -W/2 \), or outputs a negative-weight cycle in \( G \). The algorithm runs in expected time \( O\left( (m + \log\log n) \log n \log \log n \right) \) and succeeds with high probability.
\end{restatable}

To prove this theorem, we establish several intermediate results. 
\begin{claim}\label{claim:top}
    Let $(H,d)$ be any node in the decomposition tree with \( d = d_0 \). Then, for all simple paths \( P \) in \( H \) with \( w_{G'}(P) \leq 0 \), we have \( w_{G'_{\ge0}}(P) \leq d \).
\end{claim}

\begin{proof}
    Let \( \text{neg}_{G'}(P)\le 0 \) denote the total negative weight (i.e., sum of negative-weight edges) of \( P \) in \( {G'} \). For a simple path \( P \) in \( H \) with \( w_{G'}(P) \leq 0 \), each negative-weight edge in \( {G'} \) has weight at least \( -W/2 \), and \( P \) contains at most \( \lvert V(H) \rvert \leq n \) edges. Therefore, \( \text{neg}_{G'}(P) \geq -W/2 \cdot n = -d_0 \).
    In \( G'_{\ge0} \), all negative-weight edges are replaced with edges of weight 0. Hence, the weight of \( P \) in \( G'_{\ge0} \) is given by \( w_{G'_{\ge0}}(P) = w_{G'}(P) + \lvert \text{neg}_{G'}(P) \rvert \). Substituting \( w_{G'}(P) \leq 0 \) and \( \lvert \text{neg}_{G'}(P) \rvert \leq W/2 \cdot n = d \), we find \( w_{G'_{\ge0}}(P) \leq d \), as desired.
\end{proof}

\begin{claim}\label{claim:nonegcycles}
    Let $(H,d)$ be a node in the decomposition tree with \( d < d_0 \). If there exists a path \( P \) in \( H \) such that \( w_{G'}(P) \leq 0 \) but \( w_{G'_{\ge0}}(P) > d \), then  concatenating $P$ with a shortest path $P'$ in $G'_{\ge 0}$ from the end of $P$ back to its start produces a negative-weight cycle in $G$. In particular, if \( G \) has no negative-weight cycles, then for all (potentially non-simple) paths \( P \) in \( H \) with \( w_{G'}(P) \leq 0 \), we have \( w_{G'_{\ge0}}(P) \leq d \).
\end{claim}

\begin{proof}
    Suppose there exists a path \( P \) in \( H \) such that \( w_{G'}(P) \leq 0 \) but \( w_{G'_{\ge0}}(P) > d \). Let \( \text{neg}_{G'}(P) \) denote the total negative weight of \( P \) in \( H \). The weight of \( P \) in \( G'_{\ge0} \) is given by \( w_{G'_{\ge0}}(P) = w_{G'}(P) + \lvert \text{neg}_{G'}(P) \rvert \). Since \( w_{G'}(P) \leq 0 \) and \( w_{G'_{\ge0}}(P) > d \), it follows that \( \lvert \text{neg}_{G'}(P) \rvert > d \).
    For each negative-weight edge $e$ in \( H \), its weight in $G$ is $w_G(e)=w_{G'}(e)-W/2\le w_{G'}(e)-|w_{G'}(e)|$, so the weight of \( P \) in \( G \) is \( w_G(P) \le w_{G'}(P) - \lvert \text{neg}_{G'}(P) \rvert \). Substituting \( w_{G'}(P) \leq 0 \) and \( \lvert \text{neg}_{G'}(P) \rvert > d \), we find \( w_G(P) < -d \).
    By the last statement of \Cref{lem:decompose}, \( V(H) \) has diameter at most \( d \) in \( G'_{\ge0} \), so the shortest path \( P' \) in \( G'_{\ge0} \) from the end of \( P \) back to its start satisfies \( w_{G'_{\ge0}}(P') \leq d \). Since edge weights can only be smaller in $G$, we also have $w_G(P')\le d$. Combining \( P \) and \( P' \) forms a cycle \( C = P + P' \) in \( G \) with \( w_G(C) = w_G(P) + w_G(P') < -d + d = 0 \), so $C$ is a negative-weight cycle.
\end{proof}

\begin{claim}\label{claim:sparsehitting}
    Let $(H,d)$ be a non-leaf node in the decomposition tree with tag $(i,m_0)$, and let $S$ be the edges cut by \textsc{Decompose}$(H,d,i,m_0)$. Then, for any (potentially non-simple) path \( P \) in \( H \) with \( w_{G'_{\ge0}}(P) \leq d \), the number of edges of $P$ inside $S$ (counting multiplicity) has expectation at most \( O(2^i\log\log(1/\epsilon)) \).
\end{claim}

\begin{proof}
By \Cref{lem:decompose}, each edge $(u,v)$ joins $S$ with probability at most $\frac{O(2^i\log\log(1/\epsilon))}d\cdot w_{G'_{\ge0}}(u,v)$. Since $w_{G'_{\ge0}}(P)\le d$, the result follows by linearity of expectation. Note that we are ignoring the $2m\epsilon\le2m^{-4}$ probability of failure since even if all edges join $S$ in the event of failure, the overall contribution to the expectation is negligible.
\end{proof}

We now proceed to analyze the runtime and correctness of the \textsc{Scale} algorithm. We start with the leaf nodes.

\begin{lemma}
If there is an edge in a leaf node $(H,d)$ with negative weight under $w_{G'}$, then the algorithm outputs a negative-weight cycle in $G$.
\end{lemma}
\begin{proof}
Suppose that $(u,v)$ is an edge in $H$ with $w_{G'}(u,v)<0$, which means that $w_G(u,v)<-W/2$. By \Cref{lem:decompose}, $H$ has weak diameter at most $d$ under $w_{G'_{\ge0}}$, so there is a path $P$ from $v$ to $u$ with $w_G(P)\le w_{G'_{\ge0}}(P)\le d\le W/2$. Concatenating the edge $(u,v)$ with the path $P$ gives a negative-weight cycle in $G$.
\end{proof}

We now consider non-leaf nodes. Let \( n_H = |V(H)| \) and \( m_H = |E(H)| \) denote the number of vertices and edges in $H$.

\begin{lemma}\label{lemma:neg_cycle_in_G}
    Consider a non-leaf node $(H,d)$ with tag $(i,m_0)$, and suppose \( H \) does not contain a cycle with negative weight under $w_{G'}$. Then, \textsc{BellmanFordDijkstra}$(H,w_{G'}|_H,\phi_0)$ runs in expected time \(O\left( (m_H + n_H \log \log n_H)\cdot2^i\log\log(1/\epsilon) \right) \).
\end{lemma}

\begin{proof}
    In the proof, all shortest paths are with respect to edge weights $(w_{G'}|_H)_{\phi_0}$. There are two cases to consider:

    \textbf{Case 1:} For each vertex $v\in V(H)$, there exists a shortest path to $v$ with auxiliary weight at most $d$. In this case, the shortest path $P$ to $v$ with \( w_{G'_{\ge0}}(P) \leq d \) has an expected $O(2^i\log\log(1/\epsilon))$ edges inside $S_{(H,d)}$ by Claim~\ref{claim:sparsehitting}. Since $H_\phi$ has non-negative weight edges outside of $S_{(H,d)}$, we have $\mathbb E[\eta(v)]\le O(2^i\log\log(1/\epsilon))$ for the function $\eta=\eta_{(w_{G'}|_H,\,\phi_0)}$ defined in Lemma~\ref{lem:BFD}. By Lemma~\ref{lem:BFD}, \textsc{BellmanFordDijkstra} runs in expected time \( O(\sum_v(\deg_H(v)+\log\log n)\cdot\eta(v)), \) which has expectation \( O((m_H + n_H \log \log n_H) \cdot2^i\log\log(1/\epsilon)) \) since \( \mathbb E[\eta(v)]\le 2^i\log\log(1/\epsilon) \) for all $v\in V(H)$. Note that \textsc{BellmanFordDijkstra} can still terminate early and output a negative-weight cycle in $G$, but that only speeds up the running time.

    \textbf{Case 2:} For some vertex $v\in V(H)$, every shortest path to $v$ has auxiliary weight more than $d$. By Claim~\ref{claim:nonegcycles}, this implies the existence of a negative-weight cycle in \( G \). Among all shortest paths to $v$, let $Q$ be one minimizing $w_{G'_{\ge0}}(Q)$, which is still greater than $d$. Take the longest prefix $Q'$ of $Q$ with weight at most $d$ in $G'_{\ge0}$. Since $w_{G'_{\ge0}}(Q')\le d$, the number of edges of $Q'$ inside $S_{(H,d)}$ has expectation \( O(2^i\log\log(1/\epsilon)) \) by Claim~\ref{claim:sparsehitting}. Let $Q''$ be the path $Q'$ concatenated with the next edge in $Q$, so that $w_{G'_{\ge0}}(Q'')>d$ and the number of edges of $Q''$ inside $S_{(H,d)}$ still has expectation \( O(2^i\log\log(1/\epsilon)) \). Since $Q''$ is a prefix of shortest path $Q$, it is also a shortest path. Let $u\in V(H)$ be the endpoint of $Q''$, and let $i$ be the number of iterations of \textsc{BellmanFordDijkstra} before $d_i(u)= w_{G'}(Q'')$, so that $i$ has expected value $O(2^i\log\log(1/\epsilon))$. After $i$ iterations, \textsc{BellmanFordDijkstra} finds a shortest path $P$ ending at $u$ with $w_{G'}(P)=d_i(u)=w_{G'}(Q'')$. Moreover, $w_{G'_{\ge0}}(P)\ge w_{G'_{\ge0}}(Q'')$ since otherwise, replacing the prefix $Q''$ of $Q$ by $P$ results in a shortest path to $v$ of lower weight in $G'_{\ge0}$, contradicting the assumption on $Q$. It follows that after an expected $O(2^i\log\log(1/\epsilon))$ iterations of \textsc{BellmanFordDijkstra}, a path $P$ is found with $w_{G'}(P)\le 0$ and $w_{G'_{\ge0}}(P)>d$, and the algorithm terminates with a negative-weight cycle. The runtime is $O((m_H + n_H \log \log n_H) \cdot i)$, which has expectation $O((m_H + n_H \log \log n_H)\cdot2^i\log\log(1/\epsilon))$.
\end{proof}

\begin{lemma}\label{lemma:neg_cycle_in_G_prime}
Consider a non-leaf node $(H,d)$ with tag $(i,m_0)$, and suppose \( H \) contains a cycle with negative weight under $w_{G'}$. Then, \textsc{BellmanFordDijkstra}$(H,w_{G'}|_H,\phi_0)$ runs in expected time\linebreak \(O\left( (m_H + n_H \log \log n_H) \cdot2^i\log\log(1/\epsilon) \right) \).
\end{lemma}

\begin{proof}
    There are two cases to consider:

    \textbf{Case 1:} \( H \) contains a negative-weight cycle under $w_{G'}$ consisting entirely of non-positive-weight edges. All of these non-positive-weight edges under $w_{G'}$ have weight $0$ under $w_{G'_\ge0}$, and by \Cref{lem:decompose}, \textsc{Decompose} cuts these edges with probability $0$. Since this cycle is strongly connected, it must remain intact throughout the decomposition and persists in one of the leaf nodes. A negative-weight edge from this cycle is found in this leaf node, and the algorithm computes a negative-weight cycle and terminates early before $\textsc{BellmanFordDijkstra}$ is called. 

    \textbf{Case 2:} All negative-weight cycles under $w_{G'}$ include at least one positive-weight edge. This case is the most subtle and requires carefully defining a reference path $Q$.

Define $M=-(d+1)\cdot|V(H)|\cdot W$. We first claim that any (potentially non-simple) path in $H$ of weight at most $M$ under $w_{G'}$ must have weight exceeding $d$ under $w_{G'_{\ge0}}$. Given such a path, we first remove any cycles on the path of non-negative weight under $w_{G'}$. The remaining path still has weight at most $M=-(d+1)\cdot|V(H)|\cdot W$, so it must have at least $(d+1)\cdot|V(H)|$ edges, which means it can be decomposed into at least $d+1$ cycles of negative weight under $w_{G'}$. By assumption, each of these cycles has at least one positive-weight edge under $w_{G'}$. It follows that the path has weight at least $d+1$ in $G'_{\ge0}$, concluding the claim.

Among all paths in $H$ of weight at most $M$ under $w_{G'}$, let $Q$ be one with minimum possible weight under $w_{G'_{\ge0}}$, which must be greater than $d$. Take the longest prefix $Q'$ of $Q$ with weight at most $d$ under $w_{G'_{\ge0}}$. Since $w_{G'_{\ge0}}(Q')\le d$, the number of edges of $Q'$ inside $S_{(H,d)}$ has expectation \( O(2^i\log\log(1/\epsilon)) \) by Claim~\ref{claim:sparsehitting}. Let $Q''$ be the path $Q'$ concatenated with the next edge in $Q$, so that $w_{G'_{\ge0}}(Q'')>d$ and the number of edges of $Q''$ inside $S_{(H,d)}$ still has expectation \( O(2^i\log\log(1/\epsilon)) \). Let $u\in V(H)$ be the endpoint of $Q''$, and let $i$ be the number of iterations of \textsc{BellmanFordDijkstra} before $d_i(u)\le w_{G'}(Q'')$, so that $i$ has expected value $O(2^i\log\log(1/\epsilon))$. After $i$ iterations, \textsc{BellmanFordDijkstra} finds a path $P$ ending at $u$ with $w_{G'}(P)=d_i(u)\le w_{G'}(Q'')$. Moreover, $w_{G'_{\ge0}}(P)\ge w_{G'_{\ge0}}(Q'')$ since otherwise, replacing the prefix $Q''$ of $Q$ by $P$ results in a path of lower weight in $G'_{\ge0}$ and of weight at most $M$ in $H$, contradicting the assumption on $Q$. It follows that after an expected $O(2^i\log\log(1/\epsilon))$ iterations of \textsc{BellmanFordDijkstra}, a path $P$ is found with $w_{G'}(P)\le 0$ and $w_{G'_{\ge0}}(P)>d$, and the algorithm terminates with a negative-weight cycle. The runtime is $O((m_H + n_H \log \log n_H) \cdot i)$, which has expectation $O((m_H + n_H \log \log n_H)\cdot2^i\log\log(1/\epsilon))$.
\end{proof}

\begin{lemma}
The algorithm has expected running time $O((m+n\log\log n)\log n\log\log n)$.
\end{lemma}
\begin{proof}
We first bound the expected running time of each node $(H,d)$ with tag $(i,m_0)$ by\linebreak \( O((m_H + n_H \log \log n_H)\cdot 2^i\log\log(1/\epsilon)) \)   under both \textsc{Decompose} and \textsc{RecursiveDistanceComputation}. \textsc{Decompose}$(H,d,i,m_0)$ takes expected time \( O((m_H + n_H \log \log n_H)\cdot 2^i\log\log(1/\epsilon)) \) by Lemma~\ref{lem:decompose}. In \textsc{RecursiveDistanceComputation}$(H,d)$, Lemma~\ref{lem:fix-dag} adjusts $\phi$ in $O(n_H+m_H)$ time, and \Cref{lemma:neg_cycle_in_G,lemma:neg_cycle_in_G_prime} establish an expected runtime of \( O((m_H + n_H \log \log n_H) \cdot2^i\log\log(1/\epsilon)) \) for \textsc{BellmanFordDijkstra}. There may be an additional $O(m+n\log\log n)$ time to output a negative-weight cycle, but this only happens once. Overall, the expected running time is $O((m_H + n_H \log \log n_H) \cdot2^i\log\log(1/\epsilon))$.

We now bound the total expected running time by $O((m+n\log\log n)\log n\log\log n)$. For each node $(H,d)$, we charge the $O((m_H + n_H \log \log n_H) \cdot2^i\log\log(1/\epsilon))$ expected running time to the vertices in $H$ so that each vertex is charged at most $2^i\log\log(1/\epsilon)$, where a vertex $v$ with charge $c$ is responsible for $O(\deg_H(v)+\log\log n_H)$ expected running time. To bound the total expected running time, it suffices to show that each vertex is charged $O(\log n\log\log n)$ times.

For a node $(H,d)$ with tag $(i,m_0)$, define a new parameter $x=4/W\cdot dm_0$. Since the recursion becomes a leaf at $d<W/2$, we always have $d\ge W/4$, which means $x\ge1$. Also, since $d\le d_0=nW/2$, we have $x\le2nm$. Consider each possibility in \Cref{lem:decompose}:
 \begin{enumerate}
 \item If $i=0$, then
  \begin{enumerate}
  \item In the recursive call $\textsc{Decompose}(H',\Delta,0,m_{H'})$, the new parameter $x$ is at most $\frac34x$.
  \item In the recursive call $\textsc{Decompose}(H',\Delta/2,0,m_{H'})$, the new parameter $x$ is at most $\frac12x$.
  \item In the recursive call $\textsc{Decompose}(H[U],\Delta,1,m_0)$, the new parameter $i$ increases by $1$.
  \end{enumerate}
 \item If $i>0$, then
  \begin{enumerate}
  \item In the recursive call $\textsc{Decompose}(H',\Delta,0,m_{H'})$, the new parameter $x$ is at most $x/2^{2^{i-1}}$.
  \item In the recursive call $\textsc{Decompose}(H[U],\Delta,i+1,m_0)$, the new parameter $i$ increases by $1$.
  \end{enumerate}
 \end{enumerate}
This recursive structure is the same as the one in the proof of \Cref{lem:5}, except that $x$ replaces $m$ and there is a separate recursive instance for each iteration $i$. Grouping these recursive instances together and using the same analysis, we conclude that each vertex is charged $O(\log x\log\log(1/\epsilon))$ times, which is $O(\log n\log\log n)$ since $x\le 2nm$ and $\epsilon=m^{-5}$.
\end{proof}

With correctness and running time established, we now finish the proof of \Cref{thm:master}. To convert distances in \( G' \) into a reweighted graph \( G_\phi \) with edge weights at least \( -W/2 \), we note that setting \( \phi \) as the distances in $G'$ makes \( G'_\phi \) non-negative. Since \( w_{G'}(e) = w_G(e) + W/2 \) for all $e\in E$, it follows that \( w_{G_\phi}(e) = w_{G'_\phi}(e)-w_{G'}(e)+w_G(e) \geq -W/2 \), as desired.

\bibliographystyle{alpha}
\bibliography{ref2}

\appendix

\newcommand\IN{\mathit{in}}
\newcommand\OUT{\mathit{out}}
\newcommand\Bin{B^\IN}
\newcommand\Bout{B^\OUT}
\newcommand\Order{O}
\newcommand\set[1]{\{\,#1\,\}}
\newcommand\Ex{\mathbf E}
\newcommand\Geom{Geom}
\newcommand\dist{dist}

\section{Proof of Lemma~\ref{lem:BFD} (Bellman-Ford Dijkstra)}\label{app:BFD}
This section is devoted to a proof of the following lemma, stating that Dijkstra's algorithm can be adapted to work with negative edges in time depending on the $\eta_G(v)$ values. The proof is taken from~\cite{bringmann2023negative}. Recall that $\eta_G(v)$ denotes the minimum number of negative-weight edges in a shortest $s$-$v$ path in~$G$.

\BFD*

This lemma is basically \cite[Lemma~3.3]{bernstein2025negative}, but the statement differs slightly. We provide a self-contained proof that morally follows the one in~\cite[Appendix~A]{bernstein2025negative}.

We give the pseudocode for \Cref{lem:BFD} in \Cref{alg:lazy-dijkstra}. Throughout, let $G = (V, E, w)$ be the given directed weighted graph with possibly negative edge weights. We write $E^{\geq 0}$ for the subset of edges with nonnegative weight, and $E^{< 0}$ for the subset of edges with negative weight. In the pseudocode, we rely on Thorup's priority queue:

\begin{lemma}[Thorup's Priority Queue~\cite{thorup2003integer}] \label{lem:thorup}
There is a priority queue implementation for storing $n$ integer keys that supports the operations $\textsc{FindMin}$, $\textsc{Insert}$ and $\textsc{DecreaseKey}$ in constant time, and $\textsc{Delete}$ in time $\Order(\log\log n)$.
\end{lemma}

\begin{algorithm}[t]
\caption{The version of Dijkstra's algorithm implementing \Cref{lem:BFD}.} \label{alg:lazy-dijkstra}
\begin{algorithmic}[1]
\State Initialize $d[s] \gets 0$ and $d[v] \gets \infty$ for all vertices $v \neq s$
\State Initialize a Thorup priority queue $Q$ with keys $d[\cdot]$ and add $s$ to $Q$
\Repeat
\medskip
\Statex[1]\emph{(The Dijkstra phase)}
\State $A \gets \emptyset$
\While{$Q$ is nonempty}
    \State Remove the vertex $v$ from $Q$ with minimum $d[v]$
    \State Add $v$ to $A$
    \ForAll{edge $(v, x) \in E^{\geq 0}$}
        \If{$d[v] + w(v, w) < d[x]$} \label{alg:lazy-dijkstra:line:dijkstra-relax-start}
            \State Add $x$ to $Q$
            \State $d[x] \gets d[v] + w(v, x)$ \label{alg:lazy-dijkstra:line:dijkstra-relax-end}
        \EndIf
    \EndFor
\EndWhile

\medskip
\Statex[1]\emph{(The Bellman-Ford phase)}
\ForAll{$v \in A$}
    \ForAll{edge $(v, x) \in E^{< 0}$}
        \If{$d[v] + w(v, x) < d[x]$} \label{alg:lazy-dijkstra:line:bf-relax-start}
            \State Add $x$ to $Q$
            \State $d[x] \gets d[v] + w(v, x)$ \label{alg:lazy-dijkstra:line:bf-relax-end}
        \EndIf
    \EndFor
\EndFor

\medskip
\Until{$Q$ is empty}
\State\Return $d[v]$ for all vertices $v$
\end{algorithmic}
\end{algorithm}

For the analysis of the algorithm, we define two central quantities. Let $v$ be a vertex, then we define
\begin{align*}
    \dist_i(v) &= \min\set{ w(P) : \text{$P$ is an $s$-$v$-path containing less than $i$ negative edges}}, \\
    \dist'_i(v) &= \min\left\{ \dist_i(v), \min_{\substack{u \in V\\w(u, v) < 0}} \dist_i(u) + w(u, v) \right\}.
\end{align*}
Note that $\dist_0(v) = \dist'_0(v) = \infty$. We start with some observations involving these quantities $\dist_i$ and $\dist_i'$:

\begin{observation} \label{obs:1}
For all $i$, $\dist_i(v) \geq \dist'_i(v) \geq \dist_{i+1}(v)$.
\end{observation}

\begin{observation} \label{obs:2}
For all $v$,
\begin{equation*}
    \dist_{i+1}(v) = \min\left\{\dist_i(v), \min_{\substack{u \in V\\\dist_i(u) > \dist'_i(u)}} \dist'_i(u) + \dist_{G^{\geq 0}}(u, v)\right\}.
\end{equation*}
\end{observation}
\begin{proof}
The statement is clear if $\dist_{i}(v) = \dist_{i+1}(v)$, so assume that $\dist_{i+1}(v) < \dist_i(v)$. Let $P$ be the path witnessing $\dist_{i+1}(v)$, i.e., a shortest $s$-$v$-path containing less than~$i+1$ negative edges. Let $(x, u)$ denote the last negative-weight edge in $P$, and partition the path~$P$ into subpaths $P_1\, x\, u\, P_2$. Then the first segment $P_1\, x$ is a path containing less than~$i$ negative-weight edges and the segment $u\, P_2$ does not contain any negative-weight edges. Therefore,
\begin{equation*}
    \dist_{i+1}(v) = \dist_i(x) + w(x, u) + \dist_{G^{\geq 0}}(u, v) \geq \dist_i'(u) + \dist_{G^{\geq 0}}(u, v).
\end{equation*}
Suppose, for the sake of contradiction, that $\dist_i(u) = \dist'_i(u)$. Then
\begin{equation*}
    \dist_{i+1}(v) \geq \dist_i(u) + \dist_{G^{\geq 0}}(u, v) \geq \dist_i(v),
\end{equation*}
which contradicts our initial assumption.
\end{proof}

\begin{observation} \label{obs:3}
For all $v$,
\begin{equation*}
    \dist'_i(v) = \min\left\{\dist_i(v), \min_{\substack{u \in V\\\dist_{i-1}(u) > \dist_i(u)\\w(u, v) < 0}} \dist_i(u) + w(u, v) \right\}
\end{equation*}
\end{observation}
\begin{proof}
The statement is clear if $\dist_i(v) = \dist_i'(v)$, so suppose that $\dist_i'(v) < \dist_i(v)$. Then there is some vertex $u \in V$ with $w(u, v) < 0$ such that $\dist_i'(v) = \dist_i(u) + w(u, v)$. It suffices to prove that $\dist_{i-1}(u) > \dist_i(u)$. Suppose for the sake of contradiction that $\dist_{i-1}(u) = \dist_i(u)$. Then $\dist_i'(v) = \dist_{i-1}(u) + w(u, v) \geq \dist'_{i-1}(v)$, which contradicts our initial assumption (by \Cref{obs:1}).
\end{proof}

\begin{lemma}[Invariants of \Cref{alg:lazy-dijkstra}] \label{lem:lazy-dijkstra-invariants}
Consider the $i$-th iteration of the loop in \Cref{alg:lazy-dijkstra} (starting at $1$). Then the following invariants hold:
\smallskip
\begin{enumerate}[1.]
    \item After the Dijkstra phase (after Line \ref{alg:lazy-dijkstra:line:dijkstra-relax-end}):
    \begin{enumerate}[a.]
        \item $d[v] = \dist_i(v)$ for all vertices $v$, and
        \item $A = \set{v : \dist_{i-1}(v) > \dist_i(v)}$.
    \end{enumerate}
    \smallskip
    \item After the Bellman-Ford phase (after Line \ref{alg:lazy-dijkstra:line:bf-relax-end}):
    \begin{enumerate}
        \item $d[v] = \dist'_i(v)$ for all vertices $v$, and
        \item $Q = \set{ v : \dist_i(v) > \dist'_i(v) }$.
    \end{enumerate}
\end{enumerate}
\end{lemma}
\begin{proof}
We prove the invariants by induction on $i$.

\paragraph{First Dijkstra Phase.}
We start with the analysis of the first iteration, $i = 1$. The execution of the Dijkstra phase behaves exactly like the regular Dijkstra algorithm. It follows that $d[v] = \dist_{G^{\geq 0}}(s, v) = \dist_1(v)$, as claimed in Invariant 1a. Moreover, we include in $A$ exactly all vertices which were reachable from $s$ in $G^{\geq 0}$. Indeed, for these vertices $v$ we have that $\dist_1(v) = \dist_{G^{\geq 0}}(s, v) < \infty$ and $\dist_0(v) = \infty$, and thus $A = \set{ v : \dist_0(v) > \dist_1(v) }$, which proves Invariant 1b.

\paragraph{Later Dijkstra Phase.}
Next, we analyze the Dijkstra phase for a later iteration, $i > 1$. Letting $d'$ denote the state of the array $d$ after the Dijkstra phase, our goal is to prove that $d'[v] = \dist_i(v)$ for all vertices $v$. So fix any vertex $v$; we may assume that $\dist_{i+1}(v) < \dist_i(v)$, as otherwise the statement is easy using that the algorithm never increases $d[\cdot]$. A standard analysis of Dijkstra's algorithm reveals that
\begin{equation*}
    d'[v] = \min_{u \in Q} (d[u] + \dist_{G^{\geq 0}}(u, v)),
\end{equation*}
where $Q$ is the queue before the execution of Dijkstra. By plugging in the induction hypothesis and \Cref{obs:2}, we obtain that indeed
\begin{equation*}
    d'[v] = \min_{\substack{u \in V\\\dist_{i-1}(v) > \dist'_{i-1}(v)}} d[u] + \dist_{G^{\geq 0}}(u, v) = \dist_i(v),
\end{equation*}
which proves Invariant 1a.

To analyze Invariant 1b and the set $A$, first recall that we reset $A$ to an empty set before executing the Dijkstra phase. Afterwards, we add to $A$ exactly those vertices that are either (i) contained in the queue $Q$ initially or (ii) for which $d'[v] < d[v]$. Note that these sets are exactly (i) $\set{ v : \dist_i(v) > \dist'_i(v) }$ and (ii) $\set{v : \dist'_{i-1}(v) > \dist_i(v)}$ whose union is exactly $\set{v : \dist_{i-1}(v) > \dist_i(v)}$ by \Cref{obs:1}.

\paragraph{Bellman-Ford Phase.}
The analysis of the Bellman-Ford phase is simpler. Writing again~$d'$ for the state of the array $d$ after the execution of the Bellman-Ford phase, by \Cref{obs:3} we have that
\begin{equation*}
    d'[v] = \min_{\substack{u \in A\\w(u, v) < 0}} d[u] + w(u, v) = \min_{\substack{u \in V\\\dist'_{i-1}(u) > \dist_i(u)\\w(u, v) < 0}} \dist_i(u) + w(u, v) = \dist_i'(v),
\end{equation*}
which proves Invariant 2a. Here again we have assumed that $\dist_i'(v) < \dist_i(v)$, as otherwise the statement is trivial since the algorithm never increases $d[\cdot]$.

Moreover, after the Dijkstra phase has terminated, the queue $Q$ was empty. Afterwards, in the current Bellman-Ford phase, we have inserted exactly those vertices $v$ into the queue for which $\dist_i(v) > \dist_i'(v)$ and thus $Q = \set{ v : \dist_i(v) > \dist_i'(v)}$, which proves Invariant~2b.
\end{proof}

From these invariants (and the preceding observations), we can easily conclude the correctness of \Cref{alg:lazy-dijkstra}:

\begin{lemma}[Correctness of \Cref{alg:lazy-dijkstra}] \label{lem:lazy-dijkstra-correctness}
If the given graph $G$ contains a negative cycle, then \Cref{alg:lazy-dijkstra} does not terminate. Moreover, if \Cref{alg:lazy-dijkstra} terminates, then it has correctly computed $d[v] = \dist_G(s, v)$.
\end{lemma}
\begin{proof}
We show that after the algorithm has terminated, all edges $(u, v)$ are \emph{relaxed}, meaning that~$d[v] \leq d[u] + w(u, v)$. Indeed, suppose there is an edge $(u, v)$ which is not relaxed, i.e., $d[v] > d[u] + w(u, v)$. Let $i$ denote the final iteration of the algorithm. By Invariant~2a we have that $d[x] = \dist_i'(x)$ and by Invariant~2b we have that $\dist_i'(x) = \dist_i(x)$ (assuming that $Q = \emptyset$), for all vertices $x$. We distinguish two cases: If $w(u, v) \geq 0$, then we have that $\dist_i(v) > \dist_i(u) + w(u, v)$---a contradiction. And if $w(u, v) < 0$, then we have that $\dist_i'(v) = \dist_i(u) + w(u, v)$---also a contradiction.

So far we have proved that if the algorithm terminates, all edges are relaxed. It is easy to check that if $G$ contains a negative cycle, then at least one edge in that cycle cannot be relaxed. It follows that the algorithm does not terminate whenever $G$ contains a negative cycle.

Instead, assume that $G$ does not contain a negative cycle. We claim that the algorithm has correctly computed all distances. First, recall that throughout we have $d[v] \geq \dist_G(s, v)$. Consider any shortest $s$-$v$-path~$P$; we prove that $d[v] = w(P)$ by induction on the length of $P$. For $|P| = 0$, we have correctly set $d[s] = 0$ initially. (Note that $\dist_G(s, s)$ cannot be negative as otherwise $G$ would contain a negative cycle.) So assume that $P$ is nonempty and that $P$ can be written as $P_1\, u\, v$. Then by induction $d[u] = \dist_G(P_1\, u)$. Since the edge $(u, v)$ is relaxed, we have that $d[v] \leq d[u] + w(u, v) = w(P) = \dist_G(s, v)$. Recall that we also have $d[v] \geq \dist_G(s, v)$ and therefore $d[v] = \dist_G(s, v)$.
\end{proof}

For us, the most relevant change in the proof is the running time analysis. Recall that~$\eta_G(v)$ denotes the minimum number of negative edges in a shortest $s$-$v$-path, and that $\deg(v)$ denotes the out-degree of a vertex $v$.

\begin{lemma}[Running Time of \Cref{alg:lazy-dijkstra}] \label{lem:lazy-dijkstra-time}
Assume that $G$ does not contain a negative cycle. Then \Cref{alg:lazy-dijkstra} runs in time $\Order(\sum_v (\deg(v) + \log\log n) \eta_G(v))$.
\end{lemma}
\begin{proof}
Consider a single iteration of the algorithm. Letting $A$ denote the state of the set $A$ at the end of (Dijkstra's phase of) the iteration, the running time of the whole iteration can be bounded by:
\begin{equation*}
    \Order\left(\sum_{v \in A} (\deg(v) + \log\log n)\right).
\end{equation*}
Indeed, in the Dijkstra phase, in each iteration we spend time $\Order(\log\log n)$ for deleting an element from the queue (\Cref{lem:thorup}), but for each such deletion in $Q$ we add a new element to $A$. Moreover, both in the Dijkstra phase and the Bellman-Ford phase we only enumerate edges starting from a vertex in $A$, amounting for a total number of $\Order(\sum_{v \in A} \deg(v))$ edges. The inner steps of the loops (in Lines \ref{alg:lazy-dijkstra:line:dijkstra-relax-start} to \ref{alg:lazy-dijkstra:line:dijkstra-relax-end} and Lines \ref{alg:lazy-dijkstra:line:bf-relax-start} to \ref{alg:lazy-dijkstra:line:bf-relax-end}) run in constant time each (\Cref{lem:thorup}).

Let us write $A_i$ for the state of $i$ in the $i$-th iteration. Then the total running time is
\begin{equation*}
    \Order\left(\sum_{i=1}^\infty \sum_{v \in A_i} (\deg(v) + \log\log n)\right) = \Order\left(\sum_{v \in V} |\set{i : v \in A_i}| \cdot (\deg(v) + \log\log n)\right).
\end{equation*}
To complete the proof, it suffices to show that $|\set{i : v \in A_i}| \leq \eta_G(v)$. To see this, we first observe that $\dist_{\eta_G(v) + 1}(v) = \dist_{\eta_G(v) + 2} = \dots = \dist_G(s, v)$. Since, by the invariants above we know that $A_i = \set{ v : \dist_{i-1}(v) > \dist_i(v)}$, it follows that $v$ can only be contained in the sets $A_1, \dots, A_{\eta_G(v)}$.
\end{proof}

In combination, \Cref{lem:lazy-dijkstra-correctness,lem:lazy-dijkstra-time} complete the proof of \Cref{lem:BFD}.

\end{document}